\documentclass[10pt,a4paper]{article}
\usepackage[utf8]{inputenc}
\usepackage{fullpage}
\usepackage{placeins}
\usepackage[english]{babel}
\usepackage[autostyle, english = british]{csquotes}
\MakeOuterQuote{"}
\usepackage[dvipsnames]{xcolor}

\usepackage{latexsym} %Des symboles de maths en +

\usepackage{comment}

\usepackage{amssymb}
\usepackage{amsmath}
\usepackage{amsfonts, amsthm}
\usepackage{pifont}
\usepackage{enumitem}
\usepackage{tikz}
\usetikzlibrary{positioning} 
\usetikzlibrary{shapes}
\tikzstyle{S&L}=[draw,circle,minimum height=0.5cm]
\tikzstyle{vertex}=[draw,circle
,fill=black]

  \usepackage{soul} %Pour barrer en couleur 

\usepackage{algorithmicx}
\usepackage{algorithm}
\usepackage{algpseudocode}
\usepackage{hyperref}
\usepackage{cleveref}

\makeatletter
\newcommand{\algorithmiccontinue}{\textbf{continue}}
\newcommand{\Continue}{\State \algorithmiccontinue}
\makeatother

\usepackage{lscape}%Pour faire un passage en paysage 

\newcommand{\STATEnonum}{\item[]}
\newtheorem{theorem}{Theorem}[section]
\newtheorem{lemma}[theorem]{Lemma}

\newtheorem{definition}[theorem]{Definition}

\newtheorem{property}[theorem]{Proposition}

\newtheorem{claim}[theorem]{Claim}

\newlist{steps}{enumerate}{1}
\setlist[steps, 1]{label = Step \arabic*.}
\newcounter{claim}

{\noindent {}{#1}{}}{ This proves Claim~(\arabic{claim}).\vspace{1ex}}

\newenvironment{subproof}[1][\proofname]{%
  \begin{proof}[#1]%
}{%
  \end{proof}%
}

\newcommand {\sm} {\setminus}
\newcommand{\HH}{\mathcal{H}}
\newcommand{\GG}{\mathcal{G}}

\DeclareMathOperator{\free}{Free}
\DeclareMathOperator{\cw}{cw}
\renewcommand{\varnothing}{\O}

\newcommand{\func}[1]{\texttt{#1}}

\newcommand{\quotes}[1]{``#1''}

\newcommand{\textcode}{\texttt}
\newcommand{\adjMat}{\textcode{adjMat}}
\newcommand{\band}{\,\&\,}
\newcommand{\bxor}{\text{\^{}}}

\author{Cléophée Robin and Alexandre Talon}
\title{A program to find families of graphs in Free{$\{C_4,4K_1\}$} with bounded clique width}

\begin{document}
	\maketitle
\begin{abstract}
In this paper we study the class of graphs without cycles of size 4 and independent sets of size 4 as induced subgraphs: $\free\{C_4, 4K_1\}$. This is one of the three minimal minimal open cases for the complexity of the colouring problem when restricted to classes defined by excluding induced subgraphs of order 4. We investigate the clique width of some subclasses of $\free\{C_4, 4K_1\}$.

We introduce a new framework: the $(k,l,m)$-decomposition and prove that if all the graphs of a class $\cal G$ are $(k,l,m)$-decomposable, then graphs in $\cal G$ have bounded clique width. We give a few examples of such class, found with the help of a program we designed.

We also show, for any graph $G \in \free\{C_4, 4K_1\}$ that is 3 cliques coverable, an infinite family in $\free\{C_4, 4K_1\}$ of supergraphs of $G$ which have unbounded clique width.

\end{abstract}

\section{Introduction}\label{s:intro}

All graphs in this paper are finite, simple, with at least 3 vertices. As usual, the vertex set and the edge set of a graph $G$ are denoted by $V(G)$ and $E(G)$, respectively. The order of $G$ is the number of vertices of $G$. It is often denoted by $n$. 

A \emph{clique} in a graph is a set of pairwise adjacent vertices. The graph of order $n$ induced by a clique is denoted by $K_n$.  A \emph{stable set} in a graph is a set of pairwise non-adjacent vertices. The graph of order $n$ induced by a stable set is denoted by $nK_1$.  A cycle of order $n$ is denoted by $C_n$. The \emph{claw} is a graph with four vertices, one of them adjacent to the other three, and no additional edges. It is denoted by $K_{1,3}$. 

A \emph{proper colouring} of a graph G is an assignment of colours to the vertices of G in such a way that no two adjacent vertices receive the same colour. For some $k$, a \emph{$k$-colouring} of a graph is a proper colouring that uses at most $k$ colours. The \emph{xolouring problem} is that of deciding, given a graph G and an integer k, whether G admits a (proper) k-colouring. The \emph{colouring problem} is NP-complete in the general case. However, it can become polynomial in some classes of graphs. 

Given a family of graphs $\cal H$, Free$\{\cal H\}$ denotes the class of graphs without any induced subgraphs from $\cal H$. Such classes are called hereditary as, for every graph $G$ in the class, every induced subgraph of $G$ is also in the class. Given two graphs $H$ and $G$, $H\subseteq_i G$ means that $H$ is an induced subgraph of $G$. 

If one considers the complexity of the colouring problem in classes of graphs defined by forbidden induced subgraphs with at most 4 vertices, there remain three minimal open cases~\cite{vv_lozin_vertex_2015}:
\begin{itemize}
    \item Free$\{K_{1,3},4K_1\}$
    \item Free$\{K_{1,3},K_2+2K_1,4K_1\}$
    \item Free$\{C_4,4K_1\}$
\end{itemize}

In this paper we investigate the class Free$\{C_4,4K_1\}$. Several authors proved that the colouring problem is polynomial time solvable on subclasses of Free$\{C_4,4K_1\}$ defined by forbidding some extra graphs. Fraser et all. proved that the colouring problem is polynomial time solvable in Free$\{4K_1,C_4,C_5\}$~\cite{fraser_characterizations_2017}. Penev proved that the colouring problem is polynomial time solvable in Free$\{4K_1,C_4,C_6\}$~\cite{penev_coloring_2023}.

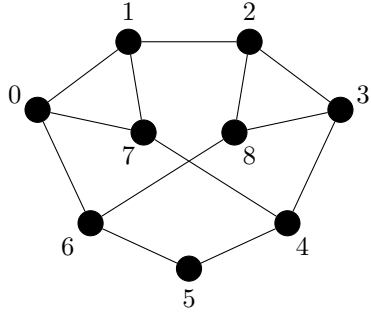
\begin{figure}
    \centering
    \begin{tikzpicture}
		
		\begin{scope}[xshift=0cm,yshift=0cm]
			\node[vertex] (a) at (-0.8,1.5){};
			\node[vertex] (b) at (0.8,1.5){};
			\node[vertex] (c) at (2,0.6){};
			\node[vertex] (d) at (1.3,-0.9){};
			\node[vertex] (e) at (0,-1.5){};
			\node[vertex] (f) at (-1.3,-0.9){};
			\node[vertex] (g) at (-2,0.6){};
			\node[vertex] (h) at (-0.6,0.3){};
            \node[vertex] (i) at (0.6,0.3){};
            \node[-] (na) at (-0.8,1.9){1};
			\node[-] (nb) at (0.8,1.9){2};
			\node[-] (nc) at (2.3,0.8){3};
			\node[-] (nd) at (1.5,-1.2){4};
			\node[-] (ne) at (0,-1.9){5};
			\node[-] (nf) at (-1.6,-1.2){6};
			\node[-] (ng) at (-2.3,0.8){0};
			\node[-] (nh) at (-0.8,0){7};
            \node[-] (ni) at (0.8,0){8};
			\draw[-] (a) -- (b);
			\draw[-] (b) -- (c);
			\draw[-] (d) -- (c);
			\draw[-] (d) -- (e);
			\draw[-] (e) -- (f);
			\draw[-] (g) -- (f);
			\draw[-] (a) -- (g);
   
			\draw[-] (a) -- (h);
			\draw[-] (g) -- (h);
			\draw[-] (d) -- (h);
   
			\draw[-] (c) -- (i);
			\draw[-] (b) -- (i);
            \draw[-] (f) -- (i);
			
		\end{scope}

	\end{tikzpicture}
    \caption{Ninjagraph $NG$}
    \label{fig:Ninjagraph}
\end{figure}

Given a family of graphs $\cal H$ and a graph $H\in {\cal H}$, ${\cal H}^H$ denotes the subclass of graphs in ${\cal H}$ that contain $H$ as induced subgraph. In other words: ${\cal H}^H=\{G\in {\cal H} : H\subseteq_i G\}$. Given a class of graphs $\cal G$, we set $\HH^{\cal G}=\cup_{G\in {\cal G}}\HH^G$.

We consider the following question: How to find graphs $H\in Free\{C_4,4K_1\}$ such that $ Free\{C_4,4K_1\}^H $ are polynomial time colourable?

Given three integers $k$, $l$, $m$ we define the $(k,l,m)$-decomposition for classes of graphs. We prove that classes of graphs that admit a $(k, l, m)$-decomposition (for some $k, l$ and $m$) have bounded clique width, therefore can be coloured in polynomial time.

 We say that $G$ is \emph{magic} with respect to a class of graphs $\HH$ if there exist $k, l, m \geq 1$ such that $\HH ^ G$ admits a $(k,l,m)$-decomposition. We present a program that generates some graphs that are magic with respect to $\free \{C_4, 4K_1\}$. In particular, the Ninjagraph (denoted by $NG$, see figure~\ref{fig:Ninjagraph}) is one of the graphs we found thanks to the program. Hence: 

\begin{theorem}
 Graphs in $ Free\{C_4,4K_1\}^{NG} $ have bounded clique width and so, are colourable in polynomial time. 
\end{theorem}

The notion of clique width will be defined later in the paper. 

\subsection{Notations}
For an integer $k$, $[k]$ denote the set of all positive integers lower or equal to $k$. 

Let $u$ be a vertex of a graph $G$ and $H$ be a subgraph of $G$. The set of neighbours of $v$ in $H$ is denoted by $N_H(v)$, and $N_H[v]=N_H(v)\cup \{v\}$. 

Given a graph $G=(V,E)$ and a subsets of vertices $V'\subseteq V$, $G[V']$ denotes the subgraph of $G$ induced by $V'$. For simplicity, $G\setminus V'$ denotes the induced subgraph $G[V\setminus V']$. 

 Let $X$ and $Y$ be two disjoint sets of vertices in a graph $G$. We say that $X$ and $Y$ are \emph{complete} (denoted by $X (1) Y$) if every vertex of $X$ is adjacent to every vertex of $Y$. We say that $X$ and $Y$ are \emph{anti-complete} (denoted by $X (0) Y$) if no vertex of $X$ is adjacent to any vertex of $Y$.  We say that $X$ and $Y$ are \emph{incompatible} (denoted by $X (-) Y$) if $X$ or $Y$ is empty.  We say that $X$ and $Y$ are \emph{nested} (denoted by $X(N) Y$) if there is an order $x_1,x_2,\dots,x_{|X|}$ on the vertices of $X$ and an order $y_1,y_2,\dots,y_{|Y|}$ on the vertices of $Y$ such that for all $i<j$, $N_Y(x_i)\subseteq N_Y(x_j)$ and $N_X(y_i)\subseteq N_X(y_j)$.
 We also say that $X$ and $Y$ are \emph{properly nested} if they are nested but are neither complete nor anti-complete.

	\subsection{Clique width }
	The \emph{clique width} of a graph $G$ \index{clique width}, denoted by $\cw(G)$ is the minimum number of labels necessary to build $G$ using the following four operations: 
	\newpage
	\begin{enumerate}
		\item Create a vertex $u$ labeled with an integer $\ell$.
		\item Make the disjoint union of two already built graphs.
		\item Add all the edges between all vertices with label $i$ and all vertices with label $j$  ($i\neq j$).
		\item Relabel all vertices of label $i$ with label $j$. 
	\end{enumerate}
	
	For example, it is possible to build the graph $K_n$ using 2 labels (see Figure~\ref{f:csK4}). Create the first vertex with label 1. While there are fewer than $n$ vertices, create a new vertex with label 2, add an edge between vertices labeled 1 and vertices with label 2 and relabel vertices of label 2 with label 1. This method gives $\cw(K_n)=2$.
	
	\begin{figure}
		
		\begin{tikzpicture}
			\begin{scope}[xshift=0cm,yshift=0cm]
				\node[draw,circle] (a) at (0,0) {1};
				\node[-] (fleche) at (0,-1) {Step 1};
			\end{scope}
			\begin{scope}[xshift=1.6cm,yshift=0cm]
				\node[draw,circle] (a) at (0,0) {1};
				\node[draw,circle] (b) at (0,1) {2};
				\node[-] (fleche) at (0,-1) {Step 3};
			\end{scope}
			\begin{scope}[xshift=3.2cm,yshift=0cm]
				\node[draw,circle] (a) at (0,0) {1};
				\node[draw,circle] (b) at (0,1) {2};
				\draw[-] (a) -- (b);
				\node[-] (fleche) at (0,-1) {Step 3};
			\end{scope}
			\begin{scope}[xshift=4.8cm,yshift=0cm]
				\node[draw,circle] (a) at (0,0) {1};
				\node[draw,circle] (b) at (0,1) {1};
				\draw[-] (a) -- (b);
				\node[-] (fleche) at (0,-1) {Step 4};
			\end{scope}
			\begin{scope}[xshift=6.4cm,yshift=0cm]
				\node[draw,circle] (a) at (0,0) {1};
				\node[draw,circle] (b) at (0,1) {1};
				\node[draw,circle] (c) at (0,2) {2};
				\draw[-] (a) -- (b);
				\node[-] (fleche) at (0,-1) {Step 5};
			\end{scope}
			\begin{scope}[xshift=8cm,yshift=0cm]
				\node[draw,circle] (a) at (0,0) {1};
				\node[draw,circle] (b) at (0,1) {1};
				\node[draw,circle] (c) at (0,2) {2};
				\draw[-] (a) -- (b);
				\draw[-] (c) -- (b);
				\draw[-] (a) to[bend left=35] (c);
				\node[-] (fleche) at (0,-1) {Step 6};
			\end{scope}
			\begin{scope}[xshift=9.6cm,yshift=0cm]
				\node[draw,circle] (a) at (0,0) {1};
				\node[draw,circle] (b) at (0,1) {1};
				\node[draw,circle] (c) at (0,2) {1};
				\draw[-] (a) -- (b);
				\draw[-] (c) -- (b);
				\draw[-] (a) to[bend left=35] (c);
				\node[-] (fleche) at (0,-1) {Step 7};
			\end{scope}
			\begin{scope}[xshift=11.2cm,yshift=0cm]
				\node[draw,circle] (a) at (0,0) {1};
				\node[draw,circle] (b) at (0,1) {1};
				\node[draw,circle] (c) at (0,2) {1};
				\node[draw,circle] (d) at (0,3) {2};
				\draw[-] (a) -- (b);
				\draw[-] (c) -- (b);
				\draw[-] (a) to[bend left=35] (c);
				\node[-] (fleche) at (0,-1) {Step 8};
			\end{scope}
			\begin{scope}[xshift=12.8cm,yshift=0cm]
				\node[draw,circle] (a) at (0,0) {1};
				\node[draw,circle] (b) at (0,1) {1};
				\node[draw,circle] (c) at (0,2) {1};
				\node[draw,circle] (d) at (0,3) {2};
				\node[-] (fleche) at (0,-1) {Step 9};
				\draw[-] (a) -- (b);
				\draw[-] (c) -- (b);
				\draw[-] (a) to[bend left=35] (c);
				\draw[-] (c) -- (d);
				\draw[-] (b) to[bend left=35] (d);
				\draw[-] (a) to[bend left=40] (d);
			\end{scope}
		\end{tikzpicture}
		\caption{$\cw(K_4)=2$}\label{f:csK4}
		
	\end{figure}
	Several problems that are NP-Complete in general can be solved in polynomial time when restricted to graphs with bounded clique width. Rao proved the following.
	\begin{theorem}[\cite{rao_msol_2007}]\label{t:RAO}
		The colouring problem is polynomial time solvable in classes of graphs with bounded clique width. 
	\end{theorem}

\subsection{Outline}
In \Cref{s:klm} we will define the \emph{$(k,l,m)$-decomposition} for infinite families of graphs. We prove that any family filling in this framework has a bounded clique width. We extend this definition to families of graphs which contain a specific graph as induced subgraph. The bounded clique width also applies to these types of families.

 \Cref{s:Magicgraphs} will be devoted to finding magic graphs in Free$\{C_4,4K_1\}$. To do so, we use a program described in the same section. We also discuss the limit of such an approach. 

 \section{$(k,l,m)$-decomposition and clique width}\label{s:klm}
In this section, we introduce the central notion of \emph{$(k,l,m)$-decomposition} for a family of graphs. The set of properties we require will enable us to bound the clique width of the whole family, hence we will use it for infinite families in this paper.

After proving the clique width property, we will present a notion to encompass infinite families admitting a $(k,l,m)$-decomposition. This will be the role of \emph{magic graphs}.

The notion of $(k,l,m)$-decomposition comes from the following property: 

 \begin{property}\label{p:nested}
   Let $G$ be a graph in Free$\{C_4\}$, and $X,Y \subseteq V(G)$ such that $X \cap Y = \varnothing$.
    If $X$ and $Y$ are cliques, then $X$ and $Y$ are nested.
    Moreover, there is a vertex in $X$ that is complete to $Y$, or a vertex in $Y$ that is anti-complete to $X$.
\end{property}
\begin{proof}
    Let $G$, $X$ and $Y$ as in the statement of the property. We proceed by contradiction and assume that the conclusion of the property is not met. Hence we can find $x_1$ and $x_2$ in $X$ such that there exist $y_1\in N_Y(x_1)\sm N_Y(x_2)$ and $y_2\in N_Y(x_2)\sm N_Y(x_1)$. 
    %Indeed, if we consider an order for the vertices of $X$, since $X$ and $Y$ are not nested then there exist two consecutive vertices in this order which each has a neighbour in $Y$ that the other is not adjacent to.
    Since $Y$ is a clique, $y_1y_2\in E(G)$ and so $\{x_1,x_2,y_2,y_1\}$ induced a $C_4$ in $G$, a contradiction. Hence,  $X$ and $Y$ are nested. 
    
    Let $x$ be a vertex in X with the maximal neighbourhood in $Y$. Suppose that $x$ has a non-neighbour $y$ in $Y$. Since $X$ and $Y$ are nested, the neighbourhood in $Y$ of every vertex in $X$ is included in $N_Y{x}$. Hence, $y$ has no neighbours in $X$. Therefore, $x$ is complete to $Y$ or otherwise $y$ is anti-complete to $X$. 
\end{proof}

 \subsection{The $(k,l,m)$-decomposition}\label{d:klm-decomposition}
 
 We begin by a simple definition to be clear in the next definitions.
 %Maybe dans def et notations au début ?

 \begin{definition}
Let $X$ a set. We say that the sets $X_1, \cdots, X_p$ is an upper-partition of $G$ if the $X_i$'s are pairwise disjoint and if $\bigcup_{i=1}^p{X_i} = X$.
 \end{definition}
 
Given a graph $G$, $X_1, \cdots, X_p$ is an upper-partition of $G$ if the $X_i$'s are pairwise disjoint and if $V(G)\subseteq \bigcup_{i=1}^p{X_i}$. The upper-partition notion extends the concept of partition allowing some of the sets to be empty. It is a general framework working for an infinity family of graphs. When applying it to each specific graph of the family, then not all sets are populated, due to the specificity of the given graph.

We can now define a $(k,l,m$)-decomposition.
 \begin{definition}
 Let $G$ be a graph and $k,l,m$ be three integers. We say that $G$ admits a \emph{($k,l,m$)-decomposition} If  there exists an upper-partition of $G$ into the sets $F, X_1, \cdots, X_k$ and, for all $i\leq k$, a sub-upper-partition of $X_i$ into $l$ subsets $X_i^1$, $X_i^2$,\dots , $X_i^l$ such that $|F|=m$ and for all i $\in [k]:$

    \begin{enumerate}
        \item $X_i$ is a clique,
		\item $\forall j \in [k]$ and $\forall a,b \in [l]$, $X_i^a$ and $X_j^b$ are nested,
        \item There exists at most one index $i^*$ such that $\forall j \in [l] \sm\{i,i^*\}$ and $\forall a,b \in [l]$ $X_i^a$ and $X_j^b$ are complete, anti-complete or incompatible.\label{item_nested} 
    \end{enumerate}

    We say that such a decomposition is a \emph{$(k,l,m)$-decomposition} of $G$.
\end{definition}

 \begin{comment}
 \end{definition}
		We further say that a family of graphs $\HH$ admits a \emph{($k,l,m$)-decomposition} if there exist three integers $k,l,m \geq 1$ such that each graph $G\in \HH$ there exists a partition of its vertices into sets $F, X_1, \cdots, X_k$ (some of them possibly empty) and, for all $i\leq k$, a subpartition of $X_i$ into $l$ subsets $X_i^1$, $X_i^2$,\dots , $X_i^l$ (some of them possibly empty) such that $|F|=m$ and for all i $\in [k]$:
\end{comment}
        
        In item 3., if there exist $ a,b \in [l]$ such that $X_i^a$ and $X_{i^*}^b$ are neither complete, anti-complete or incompatible then $(i,i^*)$ is called a \emph{bad pair}.

 We can think of the $X_i^a$ as boxes subdividing the partition. The intuition is that "almost" any pair of two boxes taken from different $X_i$'s should be complete or anti-complete. The only exceptions to this rule is for the bad pairs $\{i, i^*\}$. Note that for each $X_i$ there may exist at most one such $i^*$: a set cannot belong to two bad pairs. 

Now we can adapt the property of $(k,l,m)$-decomposition for a family of graphs. The reason we introduced this concept is to apply it to infinite families of graphs, as we can see in \Cref{l:klm_and_cw}.

\begin{definition}
Let $\mathcal{H}$ be a family of graphs and $k,m,l$ be some integers. We say that $\mathcal{H}$ admits a $(k,l,m)$-decomposition if every graph in $\HH$ admits a $(k,l,m)$-decomposition.
\end{definition}

 This is a very general framework. In \Cref{s:Magicgraphs} we constrain it for our purposes. In all of our cases  the set $F$ will be common to the whole family, so as the $X_i$'s. Using the concept of upper-partition will help us a lot to unify some results by allowing to have the same upper-partition and sub-uppersubpartition, some of their members empty depending on the graph in the family which we consider.
 In this paper, we also require that within a box $X_i^a$, each vertex has the same neighbourhood in $F$. This can be expressed as: $\forall i\in [k], a\in [l]$, for any $ u,v \in X_i^a$, $N_F(u)=N_F(v)$. This is more constrained, but allows the program to be simpler and run faster.

Observe that $k, l$ and $m$ are bounds for the sizes and numbers of sets and boxes. If a family of graphs admits a $(k,l,m)$-family-decomposition, then this allows us to establish a bound on the clique width of all the graphs in the family, as stated by the following lemma.

\begin{lemma} \label{l:klm_and_cw}
Let $\HH$ be a family of graphs, and let $k,l,m \geq 1$ be some integers  such that $\HH$ admits $(k,l,m)$-family-decomposition. Then $cw(G)\leq k\cdot l+m+1$.
\end{lemma}

\begin{proof}
		Let $\HH$ be a family of graphs which admits a $(k,l,m)$-decomposition for some $k,l,m \geq 1$. Let $G = (V,E) \in \HH$ and $F, X_1,\dots X_k$ an upper-partition of $V$ with sub-upper-partition of $X_i$ denoted by $X_i^1,\dots X_i^l$ as in the definition. Let us show that we can build this graph using clique width operations with at most $k \cdot l+m+1$ labels.
		
		We will use the following labels: $\{f_1, \cdots, f_m, \lambda\} \cup \{\alpha \beta \;|\; 1 \leq \alpha \leq l \text{ and } 1 \leq \beta \leq k\}$. We begin by a claim for the bad pairs.

		\begin{claim}\label{c:cwhalf}
			If $X_i$ and $X_j$ are nested, then it is possible to build $G[X_i\cup X_j]$ using clique width operations with at most $2l+1$ labels such that every vertex in $X_\beta^\alpha$ ($\alpha \in [l], \beta \in \{i,j\}$) ends with label \quotes{$\beta \alpha$}. 
		\end{claim}
	
		\begin{subproof}
		    We prove the claim by induction on the number of vertices in $X := X_i \cup X_j$.

		    If $|X| = 1$ then the statement is obvious. Let now $X$ be of size $k > 1$. Since $X_i$ and $X_j$ are nested, by proposition~\ref{p:nested}, there exists a vertex $v$ such that either $v\in X_i$ and $v$ is complete to $X_j$  or $v\in X_j$ and $v$ is anti-complete $X_i$. Set $a \in [l]$ and $b \in \{i,j\}$ such that $v \in X_b^a$. By the induction hypothesis, we can build $X \sm v$ with the desired properties. Note that one label, call it $\lambda$, is not used at the end of the construction because we assign at most $2l$ final labels (\quotes{$\beta \alpha$} for $\beta \in \{i, j\}$ and $\alpha \in [l]$). Create the vertex $v$ with the extra label $\lambda$. Add an edges between $v$ and all vertices  with label \quotes{$b \alpha$} for $\alpha \in [l]$. We obtain that $X_b$ is a clique. 
   
            If $b=j$ then $v$ is anti-complete $X_i$ so there is no other edges to add. If  $b=i$ then $v$ is complete to $X_j$. Hence, create all edges between $v$ and vertices of $X_j$ (those of label \quotes{$j \alpha$} for $\alpha \in [l]$). 
            
            To conclude, we assign label \quotes{$\beta \alpha$} to $v$.
			\end{subproof}

		Given \cref{c:cwhalf}, we are now able build the whole graph $G$ in the following way: 
		\begin{steps}
		    
		    \item Build $G[F]$ using labels $\{f_1,f_2,\dots, f_m\}$ (one label per vertex). From now on, for every new vertex $v$ we build, we will first create all necessary edges between $v$ and $F$. This is possible because the labels $\{f_1,f_2,\dots, f_m\}$ are only used for vertices of $G[F]$. \label{item:F}
		    
		    \item For all $X_i$ not already built. \label{item:cliques}
		    \begin{itemize}
		        \item If there exists $i^*$ such that $(i,i^*)$ is a bad pair then build the two cliques $X_i$ and $X_j$ with $2l+1$ labels ($\lambda$ and the $\beta\alpha$'s), using \Cref{c:cwhalf}. %: $\lambda$ and $\{\beta \alpha : \beta\in [l], \alpha= i,i^*\}$. 
		        Do it so that, at the ends of this step $X_\beta^\alpha$ ends with label \quotes{$\beta \alpha$}.% and label $\lambda$ is not used. 
		        \label{item:nested}
		        \item Else, build $X_i$ using $l+1$ labels ($\lambda$ and $\{i \alpha : \alpha\in [l]\}$) so that $X_i^\alpha$ ends with label \quotes{$i \alpha$}.% and label $\lambda$ is not used. 
		        \label{item:notnested}
		    \end{itemize}

			\item \label{item:1ou0} For every couple of cliques $X_i^a$ and $X_j^b$ such that $(i,j)$ is not a bad pair, either $X_i^a$ and $X_j^b$ are anti-complete or $X_i^a$ and $X_j^b$ are complete. In the second case, add an edge between the vertices with label $ia$ and the vertices with label $jb$. 
			\end{steps}

         All the edges with one endpoint in $G[F]$ are created by \ref{item:F} For all $i \in [k]$, $X_i$ is a clique thanks to  \ref{item:cliques} All edges between pairs of nested cliques are built in \ref{item:nested}, All edges between complete cliques are build in  \ref{item:notnested}. No other edges are created.

Observe that label $\lambda$ is needed at each iteration of \ref{item:cliques} but is no longer in use at the end. Hence, for every \ref{item:cliques}, we create $2l$ new labels in the first case and $l$ new labels in the second. Therefore, \ref{item:cliques} use at most $kl+1$ new labels. In addition, \ref{item:F} uses $m$ new labels. The other steps do not create new labels.  Hence the whole building of $G$ uses at most $kl+m+1$ labels. 
\end{proof}

We now introduce the notion of \emph{magic} graphs.

\begin{definition}
Given a class of graphs $\HH$, We say that $G$ is magic with respect to $\HH$ if there exist $k, l, m \geq 1$ such that $\HH ^ G$ admits a $(k,l,m)$-decomposition.
\end{definition}

The following lemma is a consequence from \Cref{l:klm_and_cw} and \Cref{t:RAO} .

\begin{lemma}\label{t:magic=cool}
 If $G$ is a magic with respect to a class of graphs $\HH$ then $\HH^G$ has bounded clique width and graphs in $\HH^G$ can be coloured in Polynomial time.
\end{lemma}

In what follows, we will simply say that a graph is magic, meaning that it is magic with respect to $\free \{C_4,4K_1\}$, the class of graphs we consider in this paper.

To conclude this section, we illustrate the concept of $(k,l,m)$-decomposition with one magic graph for $\free \{C_4,4K_1\}$: the icosahedron.

\begin{figure}
    \centering
	\begin{tikzpicture}
\begin{scope}[xshift=0cm,yshift=0cm,scale=1.5]

\node[S&L] (x) at (0,-0.25) {$v_1$};
\node[S&L](v1) at (0,0.5) {$v_2$};
\node[S&L] (v2) at (1.25,-0.25) {$v_3$};
\node[S&L] (v3) at (0.75,-1) {$v_4$};
\node[S&L] (v4) at (-0.5,-1) {$v_5$};
\node[S&L] (v5) at (-1.25,-0.25) {$v_6$};
\node[S&L] (u1) at (-1.75,1) {$v_7$};
\node[S&L] (u2) at (1.75,1) {$v_8$};
\node[S&L] (u3) at (2,-0.6) {$v_9$};
\node[S&L] (u4) at (0,-1.85) {$v_{10}$};
\node[S&L] (u5) at (-2,-0.6) {$v_{11}$};
\node[S&L] (y) at (0,-2.7) {$v_{12}$};
\draw[-] (x) -- (v1);
\draw[-] (x) -- (v2);  
\draw[-] (v2) -- (v1); 
\draw[-] (x) -- (v3);  
\draw[-] (v2) -- (v3); 
\draw[-] (x) -- (v4);  
\draw[-] (v4) -- (v3);
\draw[-] (x) -- (v5);  
\draw[-] (v4) -- (v5);
\draw[-] (v1) -- (v5);
\draw[-] (u1) -- (v5);
\draw[-] (u1) -- (v1);
\draw[-] (u2) -- (v2);
\draw[-] (u2) -- (v1);
\draw[-] (u1) -- (u2);
\draw[-] (u3) -- (v2);
\draw[-] (u3) -- (v3);
\draw[-] (u3) -- (u2);
\draw[-] (u4) -- (v4);
\draw[-] (u4) -- (v3);
\draw[-] (u4) -- (u3);
\draw[-] (u5) -- (u1);
\draw[-] (u5) -- (v4);
\draw[-] (u5) -- (v5);
\draw[-] (u4) -- (u5);

\draw[-] (y) to[bend left=80] (u1);
\draw[-] (y) to[bend right=80] (u2);
\draw[-] (y) -- (u3);
\draw[-] (y) -- (u4);
\draw[-] (y) -- (u5);

\end{scope}

\end{tikzpicture}
\caption{The Icosahedron} \label{f:icosahedron}
\end{figure}
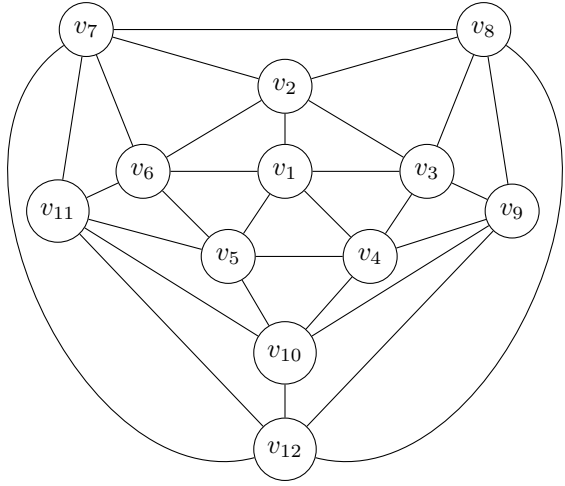

\begin{lemma}
The Icosahedron (see figure~\ref{f:icosahedron}) denoted $ICO$ is a magic graph for $\free \{C_4,4K_1\}$. Furthermore, the class $\free \{C_4,4K_1\}^{ICO}$ admits a $(4,3,12)$-decomposition. 
\end{lemma}

\begin{proof}
    The details of the proof can be found in~\cite{robin} (Lemma 5.2.4). 
    
    For all graphs $G$ in $\free \{C_4,4K_1\}$ such that $G$ contains an icosahedron $I$ as induced subgraph, it is possible to partition $V(G)\sm V(I)$ into at most 12 sets $X_1,\dots, X_k$ ($k\leq 12$) with the following properties.  For all $i,j\in [k]$: 
    \begin{itemize}
        \item $X_i$ is exactly the sets of all twins of a vertex $v_i$ of $I$. In particular ,$X_i$ is a clique. 
        \item If $v_i$, $v_j$ are two adjacent vertices of $I$, then $X_i$ and $X_j$ are complete.
        \item If $v_i$, $v_j$ are two non-adjacent vertices of $I$, then $X_i$ and $X_j$ are anti-complete. 
    \end{itemize}
    
    Now, we have that $(X_1\cup X_2\cup X_3)$, $(X_4\cup X_5\cup X_{10})$, $(X_6\cup X_7\cup X_{11})$ and $(X_8\cup X_9\cup X_{12})$ are cliques. Hence $(X_1\cup X_2\cup X_3)$, $(X_4\cup X_5 \cup X_{10})$, $(X_6\cup X_7 \cup X_{11})$ and $(X_8\cup X_9\cup X_{12})$ form a $(4,3,12)$-decomposition of $\free \{C_4,4K_1\}^{ICO}$ and so $ICO$ is a magic graph for $\free \{C_4,4K_1\}$. 
\end{proof}

By \Cref{l:klm_and_cw}, this (4,3,12)-decomposition ensures that all graphs in $\free \{C_4,4K_1\}^{ICO}$ have clique width at most 25.

 \section{Magic graphs in 
 Free$\{C_4,4K_1\}$}\label{s:Magicgraphs}

In the rest of the paper, we only consider graphs in Free$\{C_4,4K_1\}$ as our goal is to find some magic graphs in Free$\{C_4,4K_1\}$.

We will first observe that the graphs which are 3 cliques coverable (\i.e. the vertices can be partitioned into 3 cliques) cannot be magic, and then explain how to find magic graphs in practice.

 \subsection{3-clique-coverable graphs are not magic}\label{ss:3coverable_Case}

\begin{lemma}\label{l:cw_non_borne}
    Let $G$ be a graph in $Free\{C_4,4K_1\}$. \\
If $G$ is 3 cliques coverable, then $\free\{C_4,4K_1\}^G$  has unbounded clique width. 
\end{lemma}

Before proving the lemma, we recall a few known results that we shall use in proof. In order to prove \Cref{l:cw_non_borne}, we slightly modify the proof of Theorem 1 in \cite{2020_Penev}. We first need some preliminary results.
 
	By the definition of the clique width of a graph, it follows that for any pair of graphs $H$ and $G$, if $H \subseteq_i G$, then $cw(H) \leq cw(G)$. Courcelle and Olariu \cite{2000_Courcelle} established a relation between the clique width of a graph and the one of its complement.
	\begin{theorem}[\cite{2000_Courcelle}, 2000] \label{t:cw_complement}
		For any graph $G$, $cw(\overline{G}) \geq \dfrac{cw(G)}{2}$.
	\end{theorem}
	
	Brandst\"{a}dt et al.~\cite{2006_Brandstadt} used a sequence of graphs $\{G_n\}_{n=1}^{\infty}$ to show that $K_4\text{-free}$ co-chordal graphs have unbounded clique width. They define this sequence in the following way. For any integer $n\geq1$, the vertex set of $G_n$ is composed of three classes $A_n := \{a_i \ :  i \in [n] \}$, $B_n := \{b_j \ :  j \in [n]\}$, and $C_n := \{c_{i',j'} \ : i',j' \in [n]\}$, while the edge set of $G_n$ is defined as follows:
	\begin{itemize}
		\item
		$A_n$, $B_n$, and $C_n$ are stable sets
		\item
		$A_n$ is complete to $B_n$
		\item
		for all $i,i',j' \in [n]$, $a_i c_{i',j'} \in E(G_n)$ if and only if $i' \leq i$
		\item
		for all $j,i',j' \in [n]$, $b_j c_{i',j'} \in E(G_n)$ if and only if $j' \leq j$
	\end{itemize}
	
	\begin{theorem}[\cite{2006_Brandstadt}, 2006]\label{t:cw_Gn}
		For all integers $n\geq1$, $cw(G_n) \geq n$.
	\end{theorem}
	
	As in \cite{2020_Penev}, we consider a graph $G \in \mbox{ Free}\{C_4\}$ that is 3 cliques coverable and ``glue'' it to the complement of $G_n$. This allows us to construct a sequence of graphs in Free$\{C_4\}^G$ whose clique width is increasing with $n$.
	
	\begin{proof}[Proof of \Cref{l:cw_non_borne}]

		Let $G$ be a graph in Free$\{C_4,4K_1\}$ that is 3 cliques coverable.  Let $V_1,V_2,V_3$ be a partition of $V(G)$ into three cliques. 
  
		We consider the sequence of graphs $\{H_n\}_{n=1}^{\infty}$, where the vertex set is composed by four sets: $A_n := \{a_i \ :  i \in [n] \}$, $B_n := \{b_j \ :  j \in [n]\}$, $C_n := \{c_{i',j'} \ : i',j' \in [n]\}$ and $V(G)$. The edge set is defined as follows:
		\begin{itemize}
			\item
			$A_n$, $B_n$, and $C_n$ induce cliques
			\item
			$V(G)$ induces $G$
			\item
			$A_n$ is anti-complete to $B_n$
			\item
			$A_n$ is complete to $V_1$ and anti-complete to $V(G) \setminus V_1$
			\item
			$B_n$ is complete to $V_2$ and anti-complete to $V(G) \setminus V_2$
			\item
			$C_n$ is complete to $V(G)$
			\item
			for all $i,i',j' \in [n]$, $a_i c_{i',j'} \in E(H_n)$ if and only if $i' > i $
			\item
			for all $j,i',j' \in [n]$, $b_j c_{i',j'} \in E(H_n)$ if and only if $j' > j $
		\end{itemize}
		
		\begin{figure}[h]
			\centering
			\includegraphics[scale =0.9]{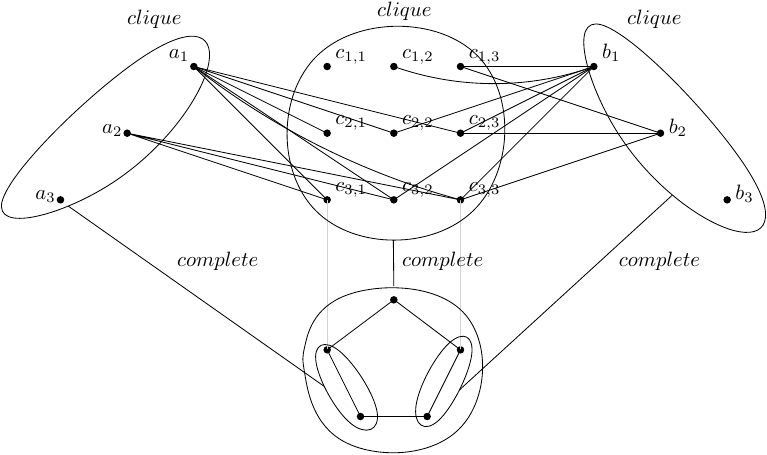}
			\caption{(Temporary) Draw of $H_3$ with $G=C_5$}
			\label{fig:H3}
		\end{figure}

		Observe that $H_n - G =\overline{G_n}$ where $G_n$ is the graph defined by Brandstädt et al.~\cite{2006_Brandstadt}. By \Cref{t:cw_complement} and \Cref{t:cw_Gn}, $ cw(H_n) \geq cw(H_n - G) = cw(\overline{G_n}) \geq \frac{n}{2} $.

        \vspace{2ex}
        
		We now show that $H_n\in \mbox{ Free}\{C_4,4K_1\}^G$.  By construction, $G\subseteq H_n$. Observe, that $H_n$ is 3 cliques coverable because $V_1\cup A_n, V_2\cup B_n, V_3\cup C_n$ is a partition of $V(H_n)$ into 3 cliques. Hence, $H_n$ does not contain any $4K_1$.

        Suppose for a contradiction that $H_n$ contains an induced $C_4$. Let $\{v_1,v_2,v_3,v_4\}$ be the vertices of the $C_4$ with $v_1v_2,v_2v_3,v_3v_4,v_4v_1\in E(G)$. Since $G\in \mbox{ Free}\{C_4\}$, $|\{v_1,v_2,v_3,v_4\}\cap V(G)|\leq 3$. Suppose, first, that $\{v_1,v_2,v_3,v_4\}\cap C_n= \varnothing$.  By the previous observation, $|\{v_1,v_2,v_3,v_4\}\cap (A_n\cup B_n)|\geq 1$. By symmetry, suppose that $v_1=a_k$, for some $k\leq n$. Observe that $N_{H_n}(v_1)\subset A_n\cup C_n\cup V_1$. Since we assume that  $\{v_1,v_2,v_3,v_4\}\cap C_n= \varnothing$, $v_2,v_4\in A_n\cup V_1$. But $A_n\cup V_1$ induces a clique, a contradiction to $v_2v_4\notin E(H_n)$.

        Hence, $\{v_1,v_2,v_3,v_4\}\cap C_n\neq \varnothing$. Suppose, by symmetry that there exist $i,j$ such that $c_{i,j} =v_1$.  Since $C_n$ is complete to $G$ and since $v_1v_3\notin E(H_n)$ we have $v_3\in A_n\cup B_n$. By symmetry, suppose that there exists $k$ such that $a_k=v_3$. Since $v_1v_3\notin E(H_n)$, $i\leq k$. Observe now that $N_{H_n}(c_{i,j})\cap N_{H_n}(a_k)\subseteq A_n\cup C_n\cup V_1$. Hence $v_2,v_4\in A_n\cup C_n\cup V_1$. Since $v_2v_4\notin E(H_n)$ and since $V_1$ is a clique complete to $A_n\cup C_n$, $v_2,v_4\in A_n\cup C_n$. Hence, there exist $i',j',k'$ such that, by symmetry, $c_{i',j'}=v_2$ and $a_{k'}=v_4$. Since $v_2v_4\notin E(H_n)$, $i'\leq k'$. Since  $v_2v_3\in E(H_n)$, $i'\geq k+1$. Hence $k'\geq k+1\geq i$ and so $v_1v_2\notin E(H_n)$, a contradiction. \\
        This concludes the proof that $H_n\in \free \{C_4,4K_1\}^G$
	\end{proof}
	
Observe that the previous result could also be derived from a result of Hoàng and Trotignon \cite{H&T}. This implies that we can focus on finding magic graphs in  Free$\{C_4,4K_1\}$  that are not 3 cliques coverable. The following is a consequence of a result of Maffray and Morel~\cite{maffray_3-colorable_2012} and will not be proved here (notations correspond to the graphs displayed on \Cref{f:graph_conj} ): 

\begin{lemma}\label{l:Interdictions}
  
  The class of graphs in Free$\{C_4\}$ that are 3 cliques coverable is exactly the class Free$\{C_4,4K_1,Ico^{-2},C_6^+,C_5+K_1,C_7,\Pi_5,F_{13}\}$.
\end{lemma}

	\begin{figure}[h]
	\centering
		\begin{tikzpicture}

\begin{scope}[xshift=-4.5cm,yshift=5cm,scale=1]

	\node[vertex] (v1) at (0,0.5){ };
	\node[vertex] (v2) at (1.1,-0.16){ };
	\node[vertex] (v3) at (0.64,-1.1){ };
	\node[vertex] (v4) at (-0.64,-1.1){ };
	\node[vertex] (v5) at (-1.1,-0.16){ };

	\node[vertex] (u1) at (-1.5,1){ };
	\node[vertex] (u2) at (1.5,1){ };
	\node[vertex] (u3) at (2,-0.75){ };
	\node[vertex] (u4) at (0,-2){ };
	\node[vertex] (u5) at (-2,-0.75){ };

	\node[-] (name) at (0,-2.75) {$ICO^{-2}$};

	\draw[-] (v2) -- (v1); 
	\draw[-] (v2) -- (v3);   
	\draw[-] (v4) -- (v3);
	\draw[-] (v4) -- (v5);
	\draw[-] (v1) -- (v5);
	\draw[-] (u1) -- (v5);
	\draw[-] (u1) -- (v1);
	\draw[-] (u2) -- (v2);
	\draw[-] (u2) -- (v1);
	\draw[-] (u1) -- (u2);
	\draw[-] (u3) -- (v2);
	\draw[-] (u3) -- (v3);
	\draw[-] (u3) -- (u2);
	\draw[-] (u4) -- (v4);
	\draw[-] (u4) -- (v3);
	\draw[-] (u4) -- (u3);
	\draw[-] (u5) -- (u1);
	\draw[-] (u5) -- (v4);
	\draw[-] (u5) -- (v5);
	\draw[-] (u4) -- (u5);

		%Ico-2
\end{scope}

\begin{scope}[xshift=0cm,yshift=0cm,scale=0.9]

	\node[vertex] (v1) at (-1.75,0){ };
	\node[vertex] (v2) at (-0.75,1.25){ };
	\node[vertex] (v3) at (0.75,1.25){ };
	\node[vertex] (v4) at (1.75,0){ };
	\node[vertex] (v5) at (0.75,-1.25){ };
	\node[vertex] (v6) at (-0.75,-1.25){ };
	\node[vertex] (x) at (0,0){ };
	
	\node[-] (name) at (0,-2) {$C_6^+$};

	\draw[-] (v1) -- (v2); 
	\draw[-] (v2) -- (v3);   
	\draw[-] (v4) -- (v3);
	\draw[-] (v4) -- (v5);
	\draw[-] (v6) -- (v5);
	\draw[-] (v1) -- (v6);
	\draw[-] (x) -- (v1);
	\draw[-] (x) -- (v4);

		%C6+
\end{scope}
	
\begin{scope}[xshift=4.25cm,yshift=4.5cm,scale=1.2]

	\node[vertex] (v1) at (0,0.5){ };
	\node[vertex] (v2) at (1.1,-0.16){ };
	\node[vertex] (v3) at (0.63,-1.2){ };
	\node[vertex] (v4) at (-0.63,-1.2){ };
	\node[vertex] (v5) at (-1.1,-0.16){ };
	\node[vertex] (x) at (0,-0.45){ };
	
	\node[-] (name) at (0,-1.8) {$C_5+K_1$};

	\draw[-] (v1) -- (v2); 
	\draw[-] (v2) -- (v3);   
	\draw[-] (v4) -- (v3);
	\draw[-] (v4) -- (v5);
	\draw[-] (v5) -- (v1);

		%C5+K_1
\end{scope}

\begin{scope}[xshift=0cm,yshift=5cm,scale=1.5]

	\node[vertex] (v1) at (-0.5,0) {};
	\node[vertex] (v2) at (0.5,0) {};
	\node[vertex] (v3) at (1,-0.5) {};
	\node[vertex] (v4) at (0.75,-1.25) {};
	\node[vertex] (v5) at (0,-1.5) {};
	\node[vertex] (v6) at (-0.75,-1.25) {};
	\node[vertex] (v7) at (-1,-0.5) {};
	\node[-] (name) at (0,-2) {$C_7$};
	\draw[-] (v1) -- (v2); 
	\draw[-] (v2) -- (v3); 
	\draw[-] (v3) -- (v4); 
	\draw[-] (v4) -- (v5);
	\draw[-] (v6) -- (v5);  
	\draw[-] (v6) -- (v7);  
	\draw[-] (v1) -- (v7);

		%C7
\end{scope}

\begin{scope}[xshift=2.75cm,yshift=-1.25cm,scale=0.6]

	\node[vertex] (v1) at (0.75,4) {};
	\node[vertex] (v2) at (3.25,4) {};
	\node[vertex] (v3) at (2,3) {};
	\node[vertex] (v4) at (0,1.75) {};
	\node[vertex] (v5) at (2,1.75) {};
	\node[vertex] (v6) at (4,1.75) {};
	\node[vertex] (v7) at (2,0.25) {};
	\node[-] (name) at (2,-0.75) {$\Pi_5$};
	\draw[-] (v1) -- (v2); 
	\draw[-] (v2) -- (v3); 
	\draw[-] (v3) -- (v1); 
	\draw[-] (v4) -- (v1);
	\draw[-] (v3) -- (v5);  
	\draw[-] (v6) -- (v2);  
	\draw[-] (v4) -- (v7);
	\draw[-] (v5) -- (v7);
	\draw[-] (v6) -- (v7);

		%Pi
\end{scope}

\begin{scope}[xshift=-6.7cm,yshift=-1.9cm,scale=0.5]

	\node[vertex] (v1) at (4,6.5) {};
	\node[vertex] (v2) at (6,6) {};
	\node[vertex] (v3) at (7,5) {};
	\node[vertex] (v4) at (7.75,3.5) {};
	\node[vertex] (v5) at (7.25,1.75) {};
	\node[vertex] (v6) at (6.25,0.35) {};
	\node[vertex] (v7) at (4.75,-0.25) {};
	\node[vertex] (v8) at (3.25,-0.25) {};
	\node[vertex] (v9) at (1.75,0.35) {};
	\node[vertex] (v10) at (0.75,1.75) {};
	\node[vertex] (v11) at (0.25,3.5) {};
	\node[vertex] (v12) at (1,5) {};
	\node[vertex] (v13) at (2,6) {};
	\node[-] (name) at (4,-1.2) {$F_{13}$};
	\draw[-] (v1) -- (v2); 
	\draw[-] (v1) -- (v3);
	\draw[-] (v1) -- (v4);
	\draw[-] (v2) -- (v3);
	\draw[-] (v2) -- (v4);
	\draw[-] (v2) -- (v5); 
	\draw[-] (v3) -- (v4); 
	\draw[-] (v3) -- (v5); 
	\draw[-] (v3) -- (v6); 
	\draw[-] (v4) -- (v5);
	\draw[-] (v4) -- (v6);
	\draw[-] (v4) -- (v7);
	\draw[-] (v5) -- (v6);
	\draw[-] (v5) -- (v7);
	\draw[-] (v5) -- (v8);  
	\draw[-] (v6) -- (v7);  
	\draw[-] (v6) -- (v8); 
	\draw[-] (v6) -- (v9); 
	\draw[-] (v7) -- (v8);
	\draw[-] (v7) -- (v9);
	\draw[-] (v7) -- (v10);
	\draw[-] (v8) -- (v9);
	\draw[-] (v8) -- (v10);
	\draw[-] (v8) -- (v11);
	\draw[-] (v9) -- (v10);
	\draw[-] (v9) -- (v11);
	\draw[-] (v9) -- (v12);
	\draw[-] (v10) -- (v11);
	\draw[-] (v10) -- (v12);
	\draw[-] (v10) -- (v13);
	\draw[-] (v11) -- (v12);
	\draw[-] (v11) -- (v13);
	\draw[-] (v11) -- (v1);
	\draw[-] (v12) -- (v13);
	\draw[-] (v12) -- (v1);
	\draw[-] (v12) -- (v2);
	\draw[-] (v13) -- (v1);
	\draw[-] (v13) -- (v2);
	\draw[-] (v13) -- (v3);

		%F13
\end{scope}
	
\end{tikzpicture}
        \caption{Graphs in $Obs$ (\Cref{l:Interdictions})}\label{f:graph_conj}
	\end{figure}

 \Cref{l:Interdictions} implies that magic graphs must contain at least one graph from $\{Ico^{-2},C_6^+,C_5+K_1,C_7,\Pi_5,F_{13}\}$. We set $Obs=\{Ico^{-2},C_6^+,C_5+K_1,C_7,\Pi_5,F_{13}\}$
	
 \subsection{Finding some magic graphs in practice}
 \label{ss:finding-magic}
We say that to vertices $u$ and $v$ are \emph{twins} in a graph if $N[u]=N[v]$. A graph is \emph{twin-free} if it does not contain any twin.

In our quest for magic graphs, we first studied some "by hand". We found a few that were interesting, but the task is long and there are a lot of graphs. So we decided to develop a program to generate the graphs in $\free \{C_4,4K_1\}$, and then test if we could automatically prove that some of them are magic. As we will describe later, the program found some magic graphs, even smaller than the one we had previously found by hand.

The first step is to generate all twin-free graphs in $\free \{C_4,4K_1\}$. In a second time, the program takes each graph $G$ generated previously and tries to find a ($k,l,m$)-decomposition for $\free \{C_4,4K_1\}^G$.  We are only interested in graphs which contain at least one graph from the set $Obs$, so we implemented this feature in our program.
Note that if the  program does not find such a partition, it does not mean that the graph is not magic.

We present here the algorithms we use, as well as some theoretical and programming optimizations we devised. Indeed, even for graphs with 15 vertices, enumerating the ones we are interested in takes a really long time if not massively parallelised nor optimised. All five algorithms can be found in \Cref{section-algos}.\\

\subsubsection{Generating graphs}\label{sss:generation}

As a preamble, let us mention that there are efficient programs to generate "small" graphs, such as Nauty and Traces. However, their preset of graph exclusions did not work for us. In order to control all the parameters, we designed a new program.

%Recall that  $Obs= \{C_5+K_1, ICO^{-2}, C_7, F_{13}, C_6^+, \Pi_5\}$. Let $H$ be a graph in $Obs$. Observe that $H$ contains a $3K_1$. Hence, for all $G\in \free \{C_4,4K_1\}^H$ and for all $v\in V(G-H)$, $v$ has at least one neighbour in $V(H)$. 
  Given an integer $n$, the generation program generates all twin-free graphs in $\free\{C_4,4K_1\}$ that contain at least one graph from $Obs$. It does so iteratively, beginning from the graphs in $Obs$ and adding one vertex at a time. Observe that there is no twin-free graphs  in $\free \{C_4,4K_1\}^{Obs}$ of order lower than 6. In addition, the only twin-free graphs  in $\free \{C_4,4K_1\}^{Obs}$ of order 6 is $C_5+K_1$.

We present here \Cref{algo-gen} (see \Cref{section-algos}), which we run iteratively. Given an integer $n$, this algorithm takes as input the set of all twin-free graphs  in $\free \{C_4,4K_1\}^{Obs}$ of order $n-1$ and output all twin-free graphs  in $\free \{C_4,4K_1\}^{Obs}$ of order $n$. To show that \Cref{algo-gen} does generate the twin-free graphs in $\free \{C_4,4K_1\}^{Obs}$, we begin by the following property.

\begin{comment}
\begin{algorithm}
%\SetAlgoLined
\textbf{Input}: $\GG$, a list of the twin-free graphs of size $n-1$\\
\textbf{Output}: A list of the twin-free graphs of  of size $n$\\

\begin{algorithmic}[1]
\State $\HH \gets [\,]$
\newline

\For{$H \in \GG$}
    \For{every possible connected graph $H'=H+v$}
        \If{$v$ has a twin \textbf{or} if $v$ is part of a $C_4$ or a $4K_1$}
         \Continue
        \EndIf
 \newline
        \If{no graphs in $\HH$ are isomorphic to $H'$}
            \State $\HH \leftarrow \HH \cup H'$
        \EndIf
    \EndFor
\EndFor   \\

\State \Return $\HH$

\end{algorithmic}

\caption{\texttt{gen\_graphs}}~\label{algo-gen}
\end{algorithm}
\end{comment}

\begin{lemma}\label{p:sous_graph_sans_jumeaux}
    Let $H$ be a twin-free graph. If $G$ is a twin-free graphs strictly containing $H$ as an induced subgraph, then  there exists $v\in V(G)$ such that $G\setminus \{v\}$ is twin-free and contains $H$ as an induced subgraph.
\end{lemma}
%TODO mise en page
\begin{proof}
    Let $G$ be a twin-free graph strictly containing $H$ as an induced subgraph. Suppose, towards a contradiction, that for all $v\in V(G)$ either $G\setminus \{v\}$ has a twin or $G\setminus \{v\}$ does not contain $H$ as induced subgraph.

    Let $a$ be a vertex such that $H$ is an induced subgraph of $G\sm \{a\}$. Amoung all the options, choose $a$ with minimum degree. Let $V'$ be a set of vertices such that $G[V']$ induces $H$ and $a\notin V'$. Since $G \sm \{a\}$ is not twin-free, there exist two vertices $b$ and $c$ such that $N_{G}[c]=N_{G}[b] \sm \{a\}: cb, ba\in E$ and $ca\notin E$.
    
    Observe first that we may assume that $G\sm \{b\}$ contains $H$ as induced subgraph. Indeed, since $H$ is twin free, either $b\notin V'$ or $c\notin V'$. If $b\notin V'$ we are done. Otherwise, $c\notin V'$. Then, since $b$ and $c$ are twins in $G \sm \{a\}$ and since $a\notin V'$, $G[V'\sm \{b\} \cup \{c\}]$ induces $H$. Up to taking $V' \setminus \{c\} \cup \{b\}$ instead of $V'$, we can assume that $b\notin V'$.
    
By our initial assumption, $G \sm \{b\}$ is not twin-free. Hence there exist two vertices $u$ and $v$ such that $N_{G}[u]=N_{G}[v]\sm \{b\}$. Therefore $uv, vb\in E$ and $ub\notin E$ and so $a,c\neq u$.
    
   Since $b$ and $c$ are twins in $G \sm\{a\}$, $cu\notin E$ and so $v\neq c$. Since $cu \notin E$ and $u$ and $v$ are twins in $G \sm\{b\}$ $cv\notin E$.  Hence since $b$ and $c$ are twins in $G \sm\{a\}$, $v=a$ (for otherwise $cv \in E$). Since $N_{G}[u]=N_{G}[a]\sm \{b\}$ and since $b\notin V'$, $G[\{a\}\cup V'\sm \{u\}]$ induces $H$. But now $d(u) = d(v)-1 = d(a)-1$. Therefore $G\sm \{u\}$ contains $H$ as induced subgraph and $u$ has degree smaller than the degree of $a$. This is a contradiction to the choice of $a$ and concludes the proof. 
\end{proof}

The following property shows that the algorithm does indeed generate all twin-free graphs in $\free \{C_4,4K_1\}^{Obs}$.

\begin{lemma}
Let $F \in Obs$ with $f$ vertices. Let $\GG_f = \{F\}$, and $\GG_{f+i}$ be the set of graphs output by $ \texttt{gen\_graphs}(\GG_{f+i-1})$ 
for $i \geq 1$. Then for $i \geq 0$, $\GG_{f+i}$ is the set of all twin-free graphs in $\free \{C_4,4K_1\}^F$ of order $f+i$.
\end{lemma}

\begin{proof}

First, by an easy induction we can see that for every $i$, every graphs in $\GG_{f+i}$ has order $f+i$, is twin-free, contains $F$, has no twins, and does not contain any $C_4$ nor any $4K_1$. The first two points are trivial (\texttt{gen\_graphs} only adds vertices). The last two points come from the fact that for each graph in $GG_{f+i-1}$ (its input), \texttt{gen\_graphs} adds a single vertex guaranteeing that this new vertex does not have twins and does not belong to any $C_4$ nor any $4K_1$. Hence from $\{F\}$, each time \texttt{gen\_graphs} is called, it maintains the fact of being twin-free and belonging to $\free \{C_4,4K_1\}$. Hence for every $i$, any graph in $\GG_{f+i}$ belongs to $\free \{C_4,4K_1\}^F$.

It remains to prove the converse: for all $i\geq 0$, every twin-free graph of order $f+i$ in $\free \{C_4,4K_1\}^F$ exists in $\GG_{f+i}$. We prove this by induction. Let $G$ be such a twin-free graph of order $f+i$ in $\free \{C_4,4K_1\}^F$. Since $F$ is the only graph of order $f$ containing $F$ as an induced subgraph, if $i=0$ then $G=F$ and the  claim holds.  Now suppose that every twin-free graph of order $f+i-1$ in $\free \{C_4,4K_1\}^F$ exists in $\GG_{f+i-1}$. By Lemma~\ref{p:sous_graph_sans_jumeaux}, there exists $v_0\in G$ such that $G\setminus\{v_0\}$ is a twin-free graph and contains $F$ as an induced subgraph. Hence, by the induction hypothesis, $G\setminus \{v_0\}$ exists in $\GG_{f+i-1}$. Step 3 of the algorithm will consider $G\setminus\{v_0\}$ and all possible ways to add a vertex to it, including adding $v_0$. So it will consider $G$. As $G$ is twin free and in $\free \{C_4,4K_1\}$, $G$ will be put in $\GG_{f+i}$.
\end{proof}

\subsubsection{Looking for magic graphs}\label{ss:magicgraph}
We recall that a graph $F$ is magic if there exist $k,l,m$ such that any graph in $\free\{C_4, 4K_1\}^F$ admits a $(k,l,m)$-decomposition. For obvious reasons, we cannot have a look at each particular graph in the infinite family $\free\{C_4, 4K_1\}^F$. So we devise a strategy to still consider them all with a finite number of computations. In the following, we will investigate, through our program, all graphs up to 18 vertices and without twins as candidates. For some, we will guarantee that they are magic. For the others, we do not say that they are not, we can only say that we do not know.
Saying this differently, we can miss some graphs which are actually magic. Let us now describe our program.

Take a graph $F$. \Cref{algo:is_magic_graph} begins by enumerating all possible neighbourhoods in $F$ for any new vertex not inducing a twin, a $C_4$ nor a $4K_1$. We denote them by $N_1, \cdots, N_p \subseteq V(F)$.
%of neighbourhoods in $F$ for any possible vertex in any possible graph $G$ that contains $F$.
Now, let $G$ be a twin-free graph in $\free\{C_4, 4K_1\}^F$.
For each $i\in [p]$ define $X_i=\{u\in G\sm F : N_F(u)=N_i\}$.
Observe that $\{F, X_1,\dots, X_p\}$ form an upper-partition of $V(G)$.

Our program then computes an array of compatibility between the $X_i$'s using \Cref{algo-relations}. This is a base level which we can refine later (recall the subdivision of the $X_i$'s in the definition of a $(k,l,m)$-decomposition).

So let us for instance consider $N_1 \neq N_2 \subseteq F$. We want to know the following: can we say that, for any $G \in \free\{C_4, 4K_1\}^F$, for any $x_1 \in V(G) \setminus F$ such that $N_F(x_1) = N_1$ and any $x_2 \in V(G) \setminus F$ such that $N_F(x_2) = N_2$, $x_1$ and $x_2$ are necessarily adjacent (1), necessarily non-adjacent (0), can be one or the other (N) or cannot both simultaneously  exist? That is to say, for a graph in $\free\{C_4, 4K_1\}^F$, does containing $x_1$ with $N_F(x_1)$ and $x_2$ with $N_F(x_2)$ force a unique relation between $x_1$ and $x_2$?
%Let $G \in \free\{C_4, 4K_1\}^F$ be \textbf{any} graph whose vertex set contains $F \cup U$  and  such that $G $, $N_F(x_1) = N_1$ and $N_F(N_2) = X_2$. We wonder if, not knowing $G$  we can know that $x_1$ and $x_2$ are adjacent (1), non-adjacent (0), can be one of the other (N), or the existence of $x_1$ and the one of $x_2$ are mutually exclusive (-).

Once this array is computed, the algorithm checks if every line of the array contains at most one occurrence of (N). It is done using \Cref{algo-find-triplets}. The latter returns a list of triplets which do not verify the property: a triplet $(A,B,C)$ is in this list if and only if $A$ and $B$ are nested, and $A$ and $C$ are nested. If the list returned by \Cref{algo-find-triplets} is empty, then \Cref{algo:is_magic_graph} declare the graph as magic. This corresponds to the definition of $(k,l,m)$-decomposition for the family of graphs $\free\{C_4, 4K_1\}^F$ if each $X_i$ is a clique.

In this paper, we only consider, for $F$, the graphs containing one of the obstructions defined in  \Cref{l:Interdictions}. Given their structures, one can check that if two vertices of $G\setminus F$ have the same neighbourhood in $F$ then they are adjacent, for otherwise $G$ would contain a $C_4$ or a $4K_1$
%for every $i$, all vertices having $N_i$ as neighbourhood in $F$ must be adjacent to each other for otherwise it induces a $C_4$ or a $4K_1$.\Cleo{plus simple d'écrire : one can check that if two vertices have the same neighbourhood in $F$, they must be adjacent for otherwise...}
Hence for all $i$,  $X_i$ is a clique. %In addition, every line of the array contains at most one occurrence of (N)
This concludes the verifications of the necessary properties to say that $F,X_1, \cdots, X_p$ form a $(p, 1, |F|)$-decomposition of any graph in $\free\{C_4, 4K_1\}^F$.

%there are at most $2^|F|$ possible \Cleo{subsets $X_i$. Hence, the size of the array is at most $2^{2|F|}$.  Checking if two sets are either (1), (0), (N) or (-) and counting the number of (N) in every line can be don in a finite number of computations, therefore the algorithm ends with a finite number of computations : ${\cal O}(2^{2|F|})$. }

%\Alex{je remplacerais par une phrase disant qu'on travaille avec un des objets finis, nb iterations finies, donc l'algo termine.}

%\Alex{faire autre sous-sections pour dire les optimisations, stratégies qu'on fait (déjà écrit  si mauvais triplet mais incompatible), aussi techniques de dire on fusionne si plusieurs N et que ça reste des cliques, en vérifiant qu'on n'ajoute pas de nouveaux N, dire aussi qu'on peut tester si on a A N B : si a et b, on peut avoir l'arête ou pas arête $\Longrightarrow$ ok, mais si a0b et on met un deuxième b (ou deuxième a): relation à tous les gens de A fixée ?)}

\vspace{2ex }
Although easy to understand, this strategy only captures one particular kind of magic graphs. We try to extend the types of graphs our program recognises as magic. For instance, it is possible that a line of the compatibility array contains two (N) but that the two corresponding sets are ($-$) to each other.
That is to say the partition contains two bad pairs $(X_1, X_2)$ and $(X_1, X_3)$ with $X_2$ and $X_3$ being incompatible: if $X_2 \neq \varnothing$ then $X_3 = \varnothing$.
%\sta{Hence, no graphs of $\free\{C_4, 4K_1\}^F$ can exist with two vertices $x_2, x_3$ such that $N_F(x_2) = N_2$ and $N_F(x_3) = N_3$.}
In this case, if $G$ is a graph in $\free\{C_4, 4K_1\}^F$, then one partition between $F,X_1, X_3, \cdots, X_p$ and $F,X_1, X_2, X_4 \cdots, X_p$ forms a $(p-1, 1, |F|)$-decomposition of $G$. Thus, by \Cref{l:klm_and_cw} the clique width of the family of graphs is still bounded by a constant.

\Cref{algo-find-triplets} looks for bad pairs. We then give them to \Cref{algo-solve-triplets} which studies them and try to merge some sets to remove the bad pairs. The idea is the following: if $(X_i, X_j)$, $(X_i, X_k)$ is a bad pair, we try to merge the cliques $X_j$ and $X_k$. This is possible if the two are complete to each other. In some cases, one merge will trigger another one. In all cases, the resulting "super-set" must be a clique. We perform this check using a classic union-find algorithm. %can lead to splitting the graphs in $\free\{C_4, 4K_1\}^F$ into a finite number of cases, each of them admitting a $(k,l,m)$-decomposition. It is done using . In case it is successful, it returns that the input graph is magic. Otherwise, it returns that it cannot conclude so.

\subsection{The algorithms}
\label{section-algos}
\begin{algorithm}[H]
%\SetAlgoLined
\textbf{Input}: $\GG$, a list of the twin-free graphs of size $n-1$\\
\textbf{Output}: A list of the twin-free graphs of  of size $n$\;
{\tiny ~\\}
\begin{algorithmic}[1]
\State $\HH \gets [\,]$
\newline

\For{$H \in \GG$}
    \For{every possible connected graph $H'=H+v$}
        \If{$v$ has a twin \textbf{or} if $v$ is part of a $C_4$ or a $4K_1$}
         \Continue
        \EndIf
 \newline
        \If{no graphs in $\HH$ are isomorphic to $H'$}
            \State $\HH \leftarrow \HH \cup H'$
        \EndIf
    \EndFor
\EndFor   \\

\State \Return $\HH$

\end{algorithmic}

\caption{\texttt{gen\_graphs}}~\label{algo-gen}
\end{algorithm}

\begin{algorithm}
%\SetAlgoLined
\textbf{Input}: A graph $G \in \GG^n$\\
\textbf{Output}: "magic" or "don't know"\\

\begin{algorithmic}[1]
\State $sets = [\,]$
\newline
\For{each subset $X$ of $[n]$}
    \State $G' \gets G + v$ such that $N(v) = X$ 
    \If{$v$ belongs to no $C_4$ and no $4K_1$} % and G+v not magic?
        \State $sets \leftarrow sets \cup X$ 
    \EndIf
\EndFor
\newline

\State $nbSets \gets |sets|$

\State $tableau \gets $ 2d array of size $nbSets\times nbSets$, filled with ' '
\newline
\For{$0 \leq i_1 < nbSets$}
    \For{$i_1 < i_2 < nbSets$}
        \State $tableau[i_1][i_2] \gets \func{get\_set\_relation}(G, sets[i_1], sets[i_2])$
        \State $tableau[i_2][i_1] \gets tableau[i_1][i_2]$
    \EndFor
\EndFor   
\newline
\State $badTriplets \gets \func{find\_bad\_triplets}(G, sets)$
%\State $trueBadTriplets \gets \func{solve\_bad\_triplets}(G, sets, BadTriplets)$
\newline
%\If{$trueBadTriplets = \varnothing$}
\If{$badTriplets = \varnothing$ \text{\textbar}\text{\textbar} \texttt{solve\_bad\_triplets}($badTriplets$)}
    \State \Return "magic"
\Else
    \State \Return "don't know"
\EndIf

\end{algorithmic}

\caption{\texttt{is\_magic\_graph}}~\label{algo-is-magic}
\label{algo:is_magic_graph}
\end{algorithm}

\begin{algorithm}
%\SetAlgoLined
\textbf{Input}: A graph $G \in \GG^n$ containing a graph in $Obs$ and two sets of vertices $X_1, X_2$\\
\textbf{Output}: "complete", "anti-complete", "incompatible" or "nested", describing the relation between the two sets\\

\begin{algorithmic}[1]
\State $G_0 \gets G + v_1 +v_2$ such that $N(v_1) = X_1$ and $N(v_2) = X_2$
\State $G_1= G_0 + \{v_1, v_2\}$
\newline
\If{$G_0$ and $G_1$ contain some $C_4$ or some $4K_1$}
%\If{$G_0 \notin \free\{C_4, 4K_1\}^{Obs}$ and $G_0 \notin \free\{C_4, 4K_1\}^{Obs}$}
    \State \Return "incompatible"
\ElsIf{$G_1$ contains some $C_4$ or some $4K_1$}
%\ElsIf{$G_1\notin \free\{C_4, 4K_1\}^{Obs}$}
    \State \Return "anti-complete"
\ElsIf{$G_0$ contains some $C_4$ or some $4K_1$}
    \State \Return "complete"
\Else
    \State \Return "nested"
\EndIf
\end{algorithmic}
\caption{\texttt{get\_sets\_relations}}
\label{algo-relations}
\end{algorithm}

\begin{algorithm}
%\SetAlgoLined
\textbf{Input}: A graph $G \in \GG^n$ containing a graph in $Obs$, and sets of vertices $Sets$ (which are computed by the \texttt{is\_magic\_graph} function)\\
\textbf{Output}: The set of bad triplets\\

\begin{algorithmic}[1]

\State $badTriplets \gets \varnothing$
\newline
\For{$X, Y, Z \in Sets$ such that $X (N)\, Y$ and $X (N)\, Z$}
    \State $biggerGraphsList = \{ G' = G + x + y +z \;|\; N'(x) = X,\; N'(y) = Y, N'(z) = Z  \text{ and }\newline \hspace*{3.9cm} G' \text{ has no } C_4 \text{ and no } 4K_1\}$ \quad// Note here that $G' $ contains a graph in $Obs$.
    \State $X_0 \gets \{G' \in biggerGraphsList \;|\; y'_b $ is not connected to $z'_b\}$
    \State $X_1 \gets \{G' \in biggerGraphsList \;|\; y'_b $ is connected to $z'_b\}$
    \newline
    \If{$X_0 \neq \varnothing $ et $ X_1 \neq \varnothing$}
       \State $badTriplets \gets badTriplets \cup (X,Y,Z)$
    \EndIf

\EndFor
\newline
\State \Return $badTriplets$

\end{algorithmic}

\caption{\texttt{find\_bad\_triplets}}~\label{algo-find-triplets}
\end{algorithm}

\begin{algorithm}
%\SetAlgoLined
\textbf{Input}: A graph $G \in \GG^n$ a set $badTriplets$ of $p$ triplets\\
\textbf{Output}: True if we manage to solve all bad triplets, false otherwise\\

\begin{algorithmic}[1]

\State $uf$ $\leftarrow$ empty union-find array of size $p$
\State $superSets \leftarrow$ array of size $p$, containing empty sets
\newline
\For{$i = 0$ to $p-1$} \quad\quad// Initialisations
        \State $superSets[i] \leftarrow \{i\}$
        \State $uf[i] = i$
\EndFor
\newline
\For{$(i,j,k) \in badTriplets$} \quad\quad// $(X_i, X_j, X_k)$ is a bad triplet in the partition.
        \State $repr1 \leftarrow$ $uf$.find($j$)
        \State $repr2 \leftarrow$ $uf$.find($k$)
        \State $uf$.union($repr1, repr2$)\\
        
        \State $superSets[repr1] \leftarrow superSets[repr1] \cup superSets[repr2]$
        \State $superSets[repr2] \leftarrow \varnothing$
\EndFor
\newline    
\For{$i = 0$ to $p-1$}
    \If{$superSets[i]$ contains two sets that are not complete to each other}
        \State \Return false
    \EndIf
\EndFor
\newline
\State \Return true
    %\For{$Y'\in Part(Y), Z'\in Part(Z), Y'_b, Z'_b$ such that $Y' symN Y'_b$ and $Z' symN Z'_b$}
    %    \State $biggerGraphsList \gets \{Gbig = G + y' + z' + y'_b + z'_b |Gbig \text{is valid and } y' \in Y'$ and same for $z', y'_b, z'_b$\}
        %\State // $biggerGraphsList$ may contain several graphs: the adjacencies between $y,' z', y'_b $ and $ z'_b$ are free.
        %\State $n_0 \gets |\{G' \in biggerGraphsList \,|\, y'_b $ is not connected to $z'_b\}|$
        %\State $n_1 \gets |\{G' \in biggerGraphsList \,|\, y'_b $ is connected to $z'_b\}|$
        %\If{$n_0 = 4 $ or $ n_1 = 4$ or ($n_0 =3$ and $n_1 = 3$)}
        %\State $X_0 \gets \{G' \in biggerGraphsList \,|\, y'_b $ is not connected to $z'_b\}$
        %\State $X_1 \gets \{G' \in biggerGraphsList \,|\, y'_b $ is connected to $z'_b\}$
       % \If{$X_0 \neq \varnothing $ et $ X_1 \neq \varnothing$}
        %\If{conditions a voir, si y'a n0 au plus 2 et lautre au plus 1 c'est ok, si un est superieur a 3 c'est pas bon. si 2 et 2 a voir si ca se combine bien}

        %\EndIf
   % \EndFor
%\EndFor

%\State \Return $trueBadTriplets$

\end{algorithmic}

\caption{\texttt{solve\_bad\_triplets}}~\label{algo-solve-triplets}
\end{algorithm}
\FloatBarrier

 \subsection{Results}
    We ran the program to generate and test graphs up to 21 vertices. We managed to generate graphs with 22 vertices (there are around 352M of these), but testing them for magic graphs would take 30 days, so we did not run the test for $n=22$. Also, we give the exact number of graphs in the Appendix (see \Cref{table:number-graphs-gen}.
    
    Observe that if every graph that contains a graph $G$ as an induced sugraph is magic, then the graph is magic. Hence, to look for magic graph of order $n$ we consider all already found magic graphs of order $n+1$ as obstruction. 
    
    Table~\ref{table:result} summarises the number of graphs output by the program.
   \begin{table}[H]
   \centering
    \begin{tabular}{|p{10ex}|p{11ex}|p{15ex}|p{24ex}|}
    \hline
    Order of the graph &\# Graphs generated (twin-free) &\# Magic graphs found (twin-free) &\# Minimal magic graphs \\
    \hline
    6& 1& 0& 0\\
    \hline
    7& 8& 0& 0\\
    \hline
    8& 34& 0& 0\\
    \hline
    9& 147& 1& 1\\
    \hline
    10& 539& 16& 9\\
    \hline
    11& 1649& 71& 5\\
    \hline
    12& 4535& 130& 0\\
    \hline
    13& 12.3k& 143& 0\\
    \hline
    14& 34.2k& 151& 0\\
    \hline
    15& 98.6k& 177& 0\\
    \hline
    16& 293k& 205& 0\\
    \hline
    17& 901k& 233& 0\\
    \hline
    18& 2.84M& 261& 0\\
    \hline
    19& 9.17M& 289& 0\\
    \hline
    20& 30.3M& 317& 0\\
    \hline
    21& 102M& 345& 0\\
    \hline
    %22& 352M& {\color{green}373? puis 401 ???}& {\color{green}XXXXXX}\\
    \end{tabular}
    \caption{Numbers of graphs output by the program}\label{table:result}
    \end{table}

    From Table~\ref{table:result}, note that the smallest magic graph output by the program is the \textit{ninjagraph} ($NG$).

    \begin{lemma}
    The class of $\free\{C_4,4K_1\}^{NG}$ has clique width at most 37. 
    \end{lemma}
    \begin{proof}

        Let $G$ be a graph in $\free\{C_4,4K_1\}^{NG}$. Let $S$ be a subset of vertices of $G$ such that $S$ induces a $NG$. We number the vertices in $S$ from 0 to 8 as in Figure~\ref{fig:Ninjagraph}. 

%        It is long, but easy to check that for any vertex $v$ in $V(G)\sm s$, $N_S(v)$ is equal to one of the following set: 

%\begin{itemize}
%   \item $X_A=\{v : N_S(v)=\{2,3,4,5\}\}$,
%   \item $X_B=\{v : N_S(v)=\{0,1,5,6\}\}$,
%   \item $X_C=\{v : N_S(v)=\{4,5,6\}\}$,
%   \item $X_D=\{v : N_S(v)=\{0,1,2,7\}\}$,
%   \item $X_E=\{v : N_S(v)=\{0,1,4,7\}\}$,
%   \item $X_F=\{v : N_S(v)=\{0,1,2,3,4,7\}\}$,
%   \item $X_G=\{v : N_S(v)=\{1,4,5,7\}\}$,
%   \item $X_H=\{v : N_S(v)=\{3,4,5,7\}\}$,
%   \item $X_I=\{v : N_S(v)=\{1,2,3,4,5,7\}\}$,
%   \item $X_J=\{v : N_S(v)=\{0,1,6,7\}\}$,
%   \item $X_K=\{v : N_S(v)=\{0,4,5,6,7\}\}$,
%   \item $X_L=\{v : N_S(v)=\{0,1,4,5,6,7\}\}$,
%   \item $X_M=\{v : N_S(v)=\{1,2,3,8\}\}$,
%   \item $X_N=\{v : N_S(v)=\{2,3,4,8\}\}$,
%   \item $X_O=\{v : N_S(v)=\{2,3,6,8\}\}$,
%   \item $X_P=\{v : N_S(v)=\{0,1,2,3,6,8\}\}$,
%   \item $X_Q=\{v : N_S(v)=\{0,5,6,8\}\}$,
%   \item $X_R=\{v : N_S(v)=\{2,5,6,8\}\}$,
%   \item $X_S=\{v : N_S(v)=\{0,1,2,5,6,8\}\}$,
%   \item $X_T=\{v : N_S(v)=\{3,4,5,6,8\}\}$,
%   \item $X_U=\{v : N_S(v)=\{2,3,4,5,6,8\}\}$,
%   \item $X_V=\{v : N_S(v)=\{1,2,3,4,7,8\}\}$,
%   \item $X_W=\{v : N_S(v)=\{0,1,2,6,7,8\}\}$,
%   \item $X_X\{v : N_S(v)=\{0,3,4,5,6,7,8\}\}$, 
%   \item $X_Y=\{v : N_S(v)=\{0,1,2,3,4,5,6,7,8\}\}$,
%\end{itemize}

         The appendix contains the output of a call to \Cref{algo:is_magic_graph} with $NG$ as input. It gives us the right upper-partition along with the compatibility table. 
         Consider the sets defined in the appendix. Observe that the sets from $A$ to $Y$ with $\{0,1,2,3,4,5,6,7,8\}$ partition every graph containing $NG$ as induced subgraph. From the table at the end of the appendix, it is easy to see that the only couples of sets that are neither complete nor anti-complete are the following: 
        
        \begin{itemize}
            \item $D (N) K$,
            \item $D (N) L$,
            \item $M (N) T$,
            \item $M (N) U$,
        \end{itemize}

        Before concluding, one can check (again from the big table in the appendix) that the following sets are cliques:

        \begin{itemize}
            \item $X_1=A\cup C\cup I\cup Y$,
            \item $X_2=B\cup D\cup J\cup P$,
            \item $X_3=E\cup F\cup G\cup H$,
            \item $X_4=K\cup L\cup Q\cup X$,
            \item $X_5=M\cup N \cup O$,
            \item $X_6=R\cup S\cup W$,
            \item $X_7=T\cup U\cup V$,
        \end{itemize}

        The only couples that are nested are $X_2(N) X_4$ and $X_5 (N) X_7$, $(S, X_1, X_2,X_3,X_4,X_5,X_6,X_7)$ form a $(7,4,9)$-decomposition of $\free\{C_4,4K_1\}^{NG}$ and so $\free\{C_4,4K_1\}^{NG}$ has clique width at most 38 by \Cref{l:klm_and_cw}. 
        
    \end{proof}

    Checking correctly all interactions between sets "by hand" takes a really long time and errors can be made.  This is why we rely on a computer to do such repetitive work. It allows us to try exhaustively this procedure on all "small" graphs.

 \subsection{Complexity, efficiency and optimisations}
As we mentioned before, generating graphs can be very very long. Indeed, there are around 352 millions twin-free graphs of order 22 in $\free\{C_4,4K_1\}$. Furthermore, after generating the 102 millions graphs of order 21 in $\free\{C_4,4K_1\}$, we enumerate all their supergraphs of size 22. Among these, a huge number have twins, contain a $C_4$, or a $4K_1$ and we have to exclude them. Also, when generating graphs of order $n+1$ from the ones of order $n$, we generate multiple isomorphic copies. For instance, unique graphs up to isomorphism represent 10\% of the graphs we generate for size 18 and the same for size 19.

We see that to obtain the graphs of size $n+1$ we enumerate many graphs that we then exclude, many more than the ones we want to generate, that is the twin-free graphs of order $n+1$ in $\free\{C_4,4K_1\}$.
Therefore we have to carefully design our program. For instance, we already mentioned that we managed to reduce a lot the number of graphs to enumerate and examine by restricting ourselves to twin-free graphs. Now let us explain the other choices and optimisation we made.

\paragraph{Notations}
Here we use bitwise operations. We will denote by $x \band y$ the bitwise and operation
 between $x$ and $y$  and by $x \bxor y$ the bitwise exclusive or (xor) operation. For instance, $1011 \band 1101 = 1001$ and $1011 \bxor 1101 = 0110$.

\paragraph{The graph data structure}
Usually, we store graphs either by adjacency matrix or adjacency lists. The first one allows for fast adjacency checks but takes much more memory than the latter except for very dense graphs. Here, we used the second approach, but implementing in a way that takes less memory and provides better performances than the two classical methods. Recall that $n$ (resp. $m$) denotes the number of vertices (resp. edges) of a graph. In our program, a Graph object contains its number of vertices (\textcode{nbVert}), edges (\textcode{nbEdge}), and an array \textcode{adjMat} of $n$ 32-bit (or 64-bit) integers. If $n <= 32$ then each cell of the array takes 4 bytes, so storing the graph takes $4n+8$ bytes.\\

The array \textcode{adjMat}
 encodes the adjacencies: its $i$-th cell encodes, as a bitset, the list of neighbours of vertex $i$. By looking at \textcode{adjMat[$i$]} in binary, the set of indices which have value 1 is precisely the set of neighbours of vertex $i$. For instance if $\textcode{adjMat[i]} = 87 = 2^{\textbf{6}}+2^\textbf{4}+2^\textbf{2}+2^\textbf{1}+2^\textbf{0}$, then the neighbours of vertex $i$ are: 6, 4, 2, 1 and 0.

This is better than storing a list of neighbours for each vertex. However, we lose the direct access to precise list of neighbours of each vertex. We solve this problem by storing a correspondence array (\textcode{adjListGlobal} in the program) such that \textcode{adjListGlobal[$x$]} contains the (integer) list of indices of the bits equal to 1 in $x$. So $\textcode{adjToListGlobal[87]} = \{0,1,2,4,6\}$. This allows us to directly retrieve the list of neighbours for any vertex of our graphs.\\
Observe that this array \textcode{adjListGlobal} is unique and common for all graphs of the same size. Therefore the memory it uses is negligible.\\

If $G$ is of order $n$, its vertices of $G$ are labelled from 0 to $n-1$ inclusive. To generate all its supergraphs, we consider all ways to add a new vertex to it. In \textcode{adjMat}, the entry for the new vertex can be any integer vertices of $G$ are labelled from 0 to $n-1$ inclusive. So enumerating these values is simple and efficient.\\

Detecting if $x$ and $y$ are neighbours amounts to checking if $\textcode{adjMat[$x$]} \band 2^y$ is nonzero: this checks if the bit of index $y$ in \textcode{adjMat[$x$]} is 1, i.e. if $y$ is a neighbour of $x$.  The time for this check is of the same order of magnitude as the one with adjacency matrices. Adding or removing an edge between $x$ and $y$ consists in performing a xor: updating \textcode{adjMat[$x$]} to $\textcode{adjMat[$x$]} \bxor 2^y$ (resp. \textcode{adjMat[$y$]} to $\textcode{adjMat[$y$]} \bxor 2^x$).

Besides, $\textcode{adjMat}[x] \band \textcode{adjMat}[y]$ 
encodes of the set of the common neighbours between $x$ and $y$: the bit of index $i$ of this number is equal to 1 if and only if vertex $i$ is neighbour with botw $x$ and $y$. Hence $x$ and $y$ are twins if and only if $\textcode{adjMat}[x] = \adjMat[x] \bxor \adjMat[y] \bxor 2^x$ and $\adjMat[y] = \adjMat[y] \bxor \adjMat[x] \bxor 2^y$ (same set of common vertices except possibly themselves).

\paragraph{Testing if a graph contains a $C_4$ or a $4K_1$}
Let $G$ be a graph with $n$ vertices that we want to extend to a graph of order $n+1$ by adding a new vertex.
Testing naively if the new vertex induces a $C_4$ takes a significant part of the running time. %Checking for $4K_1$ takes less running time, partly due to the fact.
To decrease the cost of this step to a few bitwise operations we perform some precomputations on $G$. 

Let $t$ be the new vertex we added to $G$, with adjacency \adjMat[$t$]. $t$ creates a $C_4$ if and only if there exists
 three vertices $u, v$ and $w$ such that $uvw$ is an induced $P_3$, $x$ is neighbour with $u$ and $w$ but not with $v$. This is the case if and only if $\adjMat[t] \band 2^u = 2^u$, $\adjMat[t] \band 2^w = 2^w$ and $\adjMat[t] \band 2^v = 0$, i.e. when $\adjMat[t] \band (2^u+2^v+2^w) = 2^u+2^w)$.
We do the following precomputation: we compute the list of the couples $(2^u + 2^v + 2^w, 2^u+2^w)$ for every induced $P_3$ $uvw$ Then $t$ creates a $C_4$ if and only if there exists such a couple $(c_1, c_2)$ such that $\adjMat[t] \band c_1 = c_2$.
When creating a new vertex, we check its adjacency against this list to find if it induces a $C_4$.

We can also test for a $4K_1$ by storing a list of $2^u+2^v+2^w$ for all pairwise non-adjacent triplets $(u,v,w)$ in $G$. Now $t$ induces an induced $4K_1$ if and only if $\adjMat[t] \band 2^u+2^v+2^w = 0$ for such a triplet. For the moment, detecting $4K_1$'s takes very little time, so this optimisation is not needed.

\paragraph{Testing isomorphism}
Testing isomorphism is a problem with a quasi-polynomial complexity, for general graphs.
However, here we do not consider large graphs, only  graphs with around 30 vertices at most. So the optimal algorithm would not perform well, but testing all permutations of the vertices would take a lot of time. We chose to use use a Weisfeiler-Lehman strategy \cite{weisfeiler1968reduction}: given a graph, we compute a "colour" for each of its vertices. This colour is a number which encodes some of the neighbourhood at distance $k$, for a small value of $k$ (3 for instance). We begin by giving each vertex its degree as colour: this is wave 0. Then we proceed perform new waves: for wave $i$, a vertex $u$ gets an aggregated hash of the colours its neighbours were given in wave $i-1$. The program sorts those colours and computes a non-associative hash function of the sorted list of colours.
Now, let $G$ and $H$ be two isomorphic graphs, and let $f$ be a matching of the vertices from $G$ to $H$. For each $v \in V$, the hash of $v$ in $G$ is be equal to the hash of $f(v)$ in $H$ because the hash of a vertex does not depend on the labelling of the vertices. We use this invariant to restrict the number of matchings we test. For a graph, the sorted list of the colours of its vertices is called the \emph{profile} of a graph.

Now, let $G$ be a newly generated graph. We want to know if it is an isomorphic copy of a graph we have already generated. We only test if for isomorphism against known graphs which have the same profile. Therefore, in the end we perform very few isomorphism tests, and these tests are quite fast (only vertices with the same colour can be matched).We use a "sparse" hash map to store the graphs according to their profiles: to an existing profile $p$ we associate the list of graphs the profile of which is $p$. So we quickly retrieve the list of graphs (with the same profile) to compare for isomorphism.

\paragraph{Parallelising}
We had at our disposal several machines with around one hundred cores, so we wanted to take the most part of parallelism.
For better performances, to avoid thread creating and destruction, we separated the tasks at the broadest possible level: instead of testing in parallel if a graph has a $C_4$ or the new possible adjacencies to add a vertex, each thread takes care of the whole process of generating the supergraphs, for a portion of the graphs of size $n-1$. When testing for magic graphs, each thrad tests a portion of the graphs for the magic property. The latter task is easy: each graph is independent from the others, and the file of magic graphs is very rarely updated. We use a mutex to avoid parallel writes in the file. No significant delay in adcquiring the lock was  observed.

However, for the generation of graphs this is not so easy. Indeed, we want to both check from a shared global list for graph isomorphism, and to add new graphs to this unique global list. At first, we tried to give to each thread the task of generating the graphs with some specific number of edges. This needs to do some computations and precomputations multiple times, but more importantly the workload is not distributed evenly. Some threads finish early because they processed fewer graphs (for big or small number of edges); others would terminate a lot of time afterwards. The running time of the threads looked like some normal distribution.

We devised a better approach as follows. The workload is distributed evenly: each thread takes care of several batches, each one consisting in generating all supergraphs of some portion of the list of graphs of size $n-1$. Unfortunately, this approach generates a lot of waiting for isomorphism checks, hence each thread is run only a few percent of the passing time. To avoid this, instead of having only one list of newly generated graphs, we use $256^2 = 65 536 $ independent lists of newly generated graphs. Given a graph $G$ with profile $p$, we compute a hash of $p$ on 2 bytes. This value is the index of the list (among the $256^2$ ones) in which $G$ is stored. In addition, we have 65 536 mutexes, each one guarding one list. This reduced the collisions for acquiring a lock a lot. All the processors are used in full capacity (instead of a few percents), plus they all terminate more or less at the same time (more or less a few seconds, for a running time of a few hours).

 \section{Conclusion}\label{s:conclusion}
 
We introduce a new framework, the $(k,l,m)$-decomposition,  and prove that if all the graphs of a class $\cal G$ are $(k,l,m)$-decomposable, then graphs in $\cal G$ have bounded clique width. We use this new framework to investigate the clique width of sub-classes of graphs in $\free\{C_4, 4K_1\}$.  We prove that no magic graph in $\free\{C_4, 4K_1\}$ are 3 cliques coverable. 

Using a program we devised, we generated twin-free graphs in the class $\free\{C_4, 4K_1\}$ up to order 22. We tested those of size at most 21, and found 2339 magic graphs for $\free\{C_4, 4K_1\}$, among which a subset of 15 minimal ones (see \Cref{table:result} for details). 
In particular we proved that the Ninja-graph (9 vertices) is a magic graph for $\free\{C_4, 4K_1\}$. It implies that the class of graphs in $\free\{C_4, 4K_1\}$ that contains the ninjagraph has bounded clique width.

 \subsection{Possible improvements}\label{ss:improvement}
 
 At this point, our program test if a graph is magic in a quite straightforward and basic way from the definition. Hence it is much likely that several magic graphs are not detected yet. When our program says that a graph is magic it is the case, but when it does not this does not mean that the graph is not magic. Thus, a main improvement would be to refine the program in order to detect other more elaborate $(k,l,m)$-decompositions. This could lead to discovering that some graphs we rejected so far are actually magic.
    
    One direction for improvement could be to find other ways to define the partition either by changing the initial way to define the division into sets that are clearly cliques, or by exploring other ways to combine cliques together such that they form a bigger cliques (and so, one of the sets of the partition). Also, at this point we do not subdivide the sets initially defined. For an example, when a sets $A$ is $N$ with a sets $B$ we could divide $A$ into two subsets $A_1$ and $A_2$ with $A_1$ containing all vertices of $A$ with at least one non-neighbour in $B$ and $A_2$ containing all vertices of $A$ that have no non-neighbours in $B$. Hence $A_2$ is complete to $B$ and also we have some additional information about the vertices in $A_1$. 
    
    Another improvement consist in increasing the number of graphs we investigate. At this point, we run our research for all graphs of order at most 21. What happens with graphs of larger order? Is the Breath First Search, the best strategy? We could try to define a Depth First Search strategy in order to consider graphs with a significant larger number of vertices. 
    
    Another way to decrease the number of graphs we have to investigate is to consider all magic graphs already found as an obstructions. We already implement this strategy without noticeable improvement but this may change with larger graphs. 
    
    To conclude this subsection, we would like to make the following conjecture. We believe that some graphs in Figure~\ref{f:graph_conj} are magic graphs.

%TODO was    To conclude this subsection, we would like to give the following naive and utopian conjecture. We believe that  all graphs of Figure~\ref{f:graph_conj} are magic graphs for $\free \{C_4,4K_1\}$ and so, that all graphs $\free \{C_4,4K_1\}$ that are not 3-cliques coverable, have bounded clique width. 
%\Alex{mouais, ou enlever naive}
    
 \subsection{Further interests}
 Observe that the definitions of the $(k,l,m)$-decomposition and of magic graphs do not depend on the family of graphs we consider. The same applies to Lemma~\ref{t:magic=cool}. Hence, it could be interesting to look for magic graphs in some other hereditary classes of graphs. The program we designed has some optimisations which are based upon the fact that we forbid $C_4$ and $4K_1$, but it is not difficult to remove these properties checks and replace them by others, depending on the family to be studied.
 
 From \Cref{l:cw_non_borne}, we may only find magic graphs in $\free \{C_4,4K_1\}$ that are not 3 cliques coverable. Hence, the question of the complexity of the colouring problem when restricted to graphs in $\free \{C_4,4K_1\}$ that are 3 cliques coverable remains open. One leads could be to consider the complexity of the colouring problem for 3-clique-coverable graphs in $\free \{C_4,4K_1\}$.

\section*{Acknowledgement}

We thank Marco Caoduro for his help in particular in showing that the family of 3-clique-coverable $C_4$-free graphs does not have bounded clique width.

Most of the computations were run on powerful work stations located in the LIP (Laboratoire de l'Informatique Parallèle) at the ENS de Lyon. We thank this laboratory for allowing us to run our program on their hardware.

\bibliography{Biblio}

@article{maffray_3-colorable_2012,
	title = {On {$3$}-Colorable {$P_5$}-Free Graphs},
	volume = {26},
	issn = {0895-4801, 1095-7146},
	url = {http://epubs.siam.org/doi/10.1137/110829222},
	doi = {10.1137/110829222},
	pages = {1682--1708},
	number = {4},
	journal = {{SIAM} Journal on Discrete Mathematics},
	author = {Maffray, Frédéric and Morel, Grégory},
	urldate = {2021-06-14},
	year = {2012},
	langid = {english},
}

@article{rao_msol_2007,
	title = {{MSOL} partitioning problems on graphs of bounded treewidth and clique-width},
	volume = {377},
	issn = {03043975},
	url = {https://linkinghub.elsevier.com/retrieve/pii/S0304397507002435},
	doi = {10.1016/j.tcs.2007.03.043},
	pages = {260--267},
	number = {1},
	journal = {Theoretical Computer Science},
	author = {Rao, Michaël},
	urldate = {2021-04-19},
	year = {2007},
	langid = {english},
}

@article{2000_Courcelle,
title = {Upper bounds to the clique width of graphs},
journal = {Discrete Applied Mathematics},
volume = {101},
number = {1},
pages = {77-114},
year = {2000},
issn = {0166-218X},
doi = {https://doi.org/10.1016/S0166-218X(99)00184-5},
url = {https://www.sciencedirect.com/science/article/pii/S0166218X99001845},
author = {Bruno Courcelle and Stephan Olariu},
}

@article{2006_Brandstadt,
author = {Brandstädt, Andreas and Engelfriet, Joost and Le, Hoang-Oanh and Lozin, Vadim},
year = {2006},
month = {01},
pages = {561-590},
title = {Clique-Width for 4-Vertex Forbidden Subgraphs},
volume = {39},
journal = {Theory of Computing Systems},
doi = {10.1007/s00224-005-1199-1}
}

@article{2020_Penev,
title = {On the clique-width of ({$4K_1,C_4,C_5,C_7$})-free graphs},
journal = {Discrete Applied Mathematics},
volume = {285},
pages = {688-690},
year = {2020},
issn = {0166-218X},
doi = {https://doi.org/10.1016/j.dam.2020.07.009},
url = {https://www.sciencedirect.com/science/article/pii/S0166218X20303577},
author = {Irena Penev}
}

@article{vv_lozin_vertex_2015,
	title = {Vertex coloring of graphs with few obstructions},
	url = {www.elsevier.com/locate/dam},
	pages = {273 280},
	journal = {Discrete Applied Mathematics},
	author = {{V.V. Lozin} and {D.S. Malyshev}},
	urldate = {2018-01-02},
	year = {2015},
}

@article{penev_coloring_2023,
	title = {Coloring ({$4K_1,C_4,C_6$})-free graphs},
	volume = {346},
	issn = {0012365X},
	url = {https://linkinghub.elsevier.com/retrieve/pii/S0012365X22004319},
	doi = {10.1016/j.disc.2022.113225},
	pages = {113225},
	number = {11},
	journal = {Discrete Mathematics},
	author = {Penev, Irena},
	urldate = {2023-11-14},
	date = {2023-11},
	year = {2023},
	langid = {english},
}

@article{fraser_characterizations_2017,
	title = {Characterizations  of ({$4K_1$} ,{$C_4$}, {$C_5$} )-free graphs},
	volume = {231},
	issn = {0166218X},
	url = {https://linkinghub.elsevier.com/retrieve/pii/S0166218X16304541},
	doi = {10.1016/j.dam.2016.08.016},
	pages = {166--174},
	journal = {Discrete Applied Mathematics},
	shortjournal = {Discrete Applied Mathematics},
	author = {Fraser, Dallas J. and Hamel, Angèle M. and Hoàng, Chính T. and Holmes, Kevin and {LaMantia}, Tom P.},
	urldate = {2021-04-29},
	date = {2017-11},
	year = {2017},
	langid = {english},
}

@article{weisfeiler1968reduction,
  title={A reduction of a graph to a canonical form and an algebra arising during this reduction},
  author={Leman, Andrei and Weisfeiler, Boris},
  journal={Nauchno-Technicheskaya Informatsiya},
  volume={2},
  number={9},
  pages={12--16},
  year={1968}
}

@phdthesis{robin,
  TITLE = {{Hereditary classes of graphs: from structure to coloring}},
  AUTHOR = {Robin, Cléophée},
  URL = {https://tel.archives-ouvertes.fr/tel-03587403},
  NUMBER = {2021GRALM041},
  SCHOOL = {{Universit{\'e} Grenoble Alpes}},
  YEAR = {2021},
  TYPE = {Thèse},
  HAL_ID = {tel-03587403},
  HAL_VERSION = {v1},
}

@article{H&T,
title = {A class of graphs with large rankwidth},
journal = {Discrete Mathematics},
volume = {347},
number = {1},
pages = {113699},
year = {2024},
issn = {0012-365X},
doi = {https://doi.org/10.1016/j.disc.2023.113699},
url = {https://www.sciencedirect.com/science/article/pii/S0012365X23003850},
author = {Chính T. Hoàng and Nicolas Trotignon}
}

\newpage
 \section{Apendix}
 \begin{center}
    \begin{minipage}{0.89\textwidth}

    This is the output of \Cref{algo:is_magic_graph} described in section~\ref{ss:magicgraph} when the input is the graph $NG$. 
    
    \textbf{Printing graph}: this is a description of $NG$, first with the number of vertices and edges, then with the list of edges. 1;0 means that vertex 1 is adjacent to vertex 0. 

\textbf{Printing neighbourhood of vertices}: this gives, for each new vertex, the list of its neighbours in $NG$. For instance, set A contains every vertex which neighbourhood in NG is exactly $\{2,3,4,5\}$. 

\textbf{Table}: for every pair of sets, specifies if they are complete ($(1)$), anti-complete ($(0)$), incompatible ($(-)$) or just nested ($(N)$). We consider that a set is complete to itmself. 
    \end{minipage}
    \end{center}

\textbf{Printing graph:}

\small{9 vertices and 13 edges.

1;0 2;1 3;2 4;3 5;4 6;0 6;5 7;0 7;1 7;4 8;2 8;3 8;6}

\textbf{Printing neighbourhood of vertices:}

\small{set A: 2, 3, 4, 5,

set B: 0, 1, 5, 6,

set C: 4, 5, 6,

set D: 0, 1, 2, 7,

set E: 0, 1, 4, 7,

set F: 0, 1, 2, 3, 4, 7,

set G: 1, 4, 5, 7,

set H: 3, 4, 5, 7,

set I: 1, 2, 3, 4, 5, 7,

set J: 0, 1, 6, 7,

set K: 0, 4, 5, 6, 7,

set L: 0, 1, 4, 5, 6, 7,

set M: 1, 2, 3, 8,

set N: 2, 3, 4, 8,

set O: 2, 3, 6, 8,

set P: 0, 1, 2, 3, 6, 8,

set Q: 0, 5, 6, 8,

set R: 2, 5, 6, 8,

set S: 0, 1, 2, 5, 6, 8,

set T: 3, 4, 5, 6, 8,

set U: 2, 3, 4, 5, 6, 8,

set V: 1, 2, 3, 4, 7, 8,

set W: 0, 1, 2, 6, 7, 8,

set X: 0, 3, 4, 5, 6, 7, 8,

set Y: 0, 1, 2, 3, 4, 5, 6, 7, 8,}

\begin{landscape}
\vspace*{1.7cm}
\begin{table}[h]

    \begin{tabular}{ccccccccccccccccccccccccccc}
&	A&	B&	C&	D&	E&	F&	G&	H&	I&	J&	K&	L&	M&	N&	O&	P&	Q&	R&	S&	T&	U&	V&	W&	X&	Y&	\\
A&	$(1)$&	$(-)$&	$(1)$&	$(0)$&	$(0)$&	$(1)$&	$(-)$&	$(1)$&	$(1)$&	$(0)$&	$(1)$&	$(0)$&	$(1)$&	$(1)$&	$(0)$&	$(0)$&	$(0)$&	$(-)$&	$(-)$&	$(-)$&	$(1)$&	$(1)$&	$(0)$&	$(-)$&	$(1)$&	\\
B&	$(-)$&	$(1)$&	$(1)$&	$(1)$&	$(0)$&	$(0)$&	$(-)$&	$(0)$&	$(-)$&	$(1)$&	$(-)$&	$(1)$&	$(0)$&	$(0)$&	$(0)$&	$(1)$&	$(1)$&	$(-)$&	$(1)$&	$(1)$&	$(0)$&	$(0)$&	$(1)$&	$(-)$&	$(1)$&	\\
C&	$(1)$&	$(1)$&	$(1)$&	$(0)$&	$(0)$&	$(0)$&	$(1)$&	$(1)$&	$(1)$&	$(0)$&	$(1)$&	$(1)$&	$(0)$&	$(0)$&	$(0)$&	$(0)$&	$(1)$&	$(1)$&	$(1)$&	$(1)$&	$(1)$&	$(0)$&	$(0)$&	$(1)$&	$(1)$&	\\
D&	$(0)$&	$(1)$&	$(0)$&	$(1)$&	$(1)$&	$(1)$&	$(1)$&	$(0)$&	$(1)$&	$(1)$&	${\color{red}(N)}$&	${\color{red}(N)}$&	$(1)$&	$(0)$&	$(0)$&	$(1)$&	$(0)$&	$(0)$&	$(1)$&	$(0)$&	$(0)$&	$(1)$&	$(1)$&	$(0)$&	$(1)$&	\\
E&	$(0)$&	$(0)$&	$(0)$&	$(1)$&	$(1)$&	$(1)$&	$(1)$&	$(1)$&	$(1)$&	$(1)$&	$(1)$&	$(1)$&	$(0)$&	$(0)$&	$(0)$&	$(0)$&	$(0)$&	$(0)$&	$(0)$&	$(0)$&	$(0)$&	$(1)$&	$(1)$&	$(1)$&	$(1)$&	\\
F&	$(1)$&	$(0)$&	$(0)$&	$(1)$&	$(1)$&	$(1)$&	$(1)$&	$(1)$&	$(1)$&	$(1)$&	$(1)$&	$(1)$&	$(1)$&	$(1)$&	$(0)$&	$(1)$&	$(0)$&	$(0)$&	$(-)$&	$(0)$&	$(-)$&	$(1)$&	$(-)$&	$(-)$&	$(1)$&	\\
G&	$(-)$&	$(-)$&	$(1)$&	$(1)$&	$(1)$&	$(1)$&	$(1)$&	$(1)$&	$(1)$&	$(0)$&	$(-)$&	$(1)$&	$(0)$&	$(0)$&	$(0)$&	$(0)$&	$(0)$&	$(-)$&	$(-)$&	$(1)$&	$(0)$&	$(1)$&	$(0)$&	$(-)$&	$(1)$&	\\
H&	$(1)$&	$(0)$&	$(1)$&	$(0)$&	$(1)$&	$(1)$&	$(1)$&	$(1)$&	$(1)$&	$(0)$&	$(1)$&	$(1)$&	$(0)$&	$(1)$&	$(0)$&	$(0)$&	$(0)$&	$(0)$&	$(0)$&	$(1)$&	$(1)$&	$(1)$&	$(0)$&	$(1)$&	$(1)$&	\\
I&	$(1)$&	$(-)$&	$(1)$&	$(1)$&	$(1)$&	$(1)$&	$(1)$&	$(1)$&	$(1)$&	$(0)$&	$(-)$&	$(1)$&	$(1)$&	$(1)$&	$(0)$&	$(-)$&	$(0)$&	$(-)$&	$(-)$&	$(-)$&	$(1)$&	$(1)$&	$(-)$&	$(-)$&	$(1)$&	\\
J&	$(0)$&	$(1)$&	$(0)$&	$(1)$&	$(1)$&	$(1)$&	$(0)$&	$(0)$&	$(0)$&	$(1)$&	$(1)$&	$(1)$&	$(0)$&	$(0)$&	$(0)$&	$(1)$&	$(1)$&	$(0)$&	$(1)$&	$(0)$&	$(0)$&	$(0)$&	$(1)$&	$(1)$&	$(1)$&	\\
K&	$(1)$&	$(-)$&	$(1)$&	${\color{red}(N)}$&	$(1)$&	$(1)$&	$(-)$&	$(1)$&	$(-)$&	$(1)$&	$(1)$&	$(1)$&	$(0)$&	$(0)$&	$(0)$&	$(0)$&	$(1)$&	$(1)$&	$(-)$&	$(1)$&	$(1)$&	$(0)$&	$(1)$&	$(1)$&	$(1)$&	\\
L&	$(0)$&	$(1)$&	$(1)$&	${\color{red}(N)}$&	$(1)$&	$(1)$&	$(1)$&	$(1)$&	$(1)$&	$(1)$&	$(1)$&	$(1)$&	$(0)$&	$(0)$&	$(0)$&	$(-)$&	$(1)$&	$(0)$&	$(1)$&	$(1)$&	$(-)$&	$(-)$&	$(1)$&	$(1)$&	$(1)$&	\\
M&	$(1)$&	$(0)$&	$(0)$&	$(1)$&	$(0)$&	$(1)$&	$(0)$&	$(0)$&	$(1)$&	$(0)$&	$(0)$&	$(0)$&	$(1)$&	$(1)$&	$(1)$&	$(1)$&	$(0)$&	$(1)$&	$(1)$&	${\color{red}(N)}$&	${\color{red}(N)}$&	$(1)$&	$(1)$&	$(0)$&	$(1)$&	\\
N&	$(1)$&	$(0)$&	$(0)$&	$(0)$&	$(0)$&	$(1)$&	$(0)$&	$(1)$&	$(1)$&	$(0)$&	$(0)$&	$(0)$&	$(1)$&	$(1)$&	$(1)$&	$(1)$&	$(0)$&	$(0)$&	$(0)$&	$(1)$&	$(1)$&	$(1)$&	$(0)$&	$(1)$&	$(1)$&	\\
O&	$(0)$&	$(0)$&	$(0)$&	$(0)$&	$(0)$&	$(0)$&	$(0)$&	$(0)$&	$(0)$&	$(0)$&	$(0)$&	$(0)$&	$(1)$&	$(1)$&	$(1)$&	$(1)$&	$(1)$&	$(1)$&	$(1)$&	$(1)$&	$(1)$&	$(1)$&	$(1)$&	$(1)$&	$(1)$&	\\
P&	$(0)$&	$(1)$&	$(0)$&	$(1)$&	$(0)$&	$(1)$&	$(0)$&	$(0)$&	$(-)$&	$(1)$&	$(0)$&	$(-)$&	$(1)$&	$(1)$&	$(1)$&	$(1)$&	$(1)$&	$(1)$&	$(1)$&	$(1)$&	$(1)$&	$(-)$&	$(1)$&	$(-)$&	$(1)$&	\\
Q&	$(0)$&	$(1)$&	$(1)$&	$(0)$&	$(0)$&	$(0)$&	$(0)$&	$(0)$&	$(0)$&	$(1)$&	$(1)$&	$(1)$&	$(0)$&	$(0)$&	$(1)$&	$(1)$&	$(1)$&	$(1)$&	$(1)$&	$(1)$&	$(1)$&	$(0)$&	$(1)$&	$(1)$&	$(1)$&	\\
R&	$(-)$&	$(-)$&	$(1)$&	$(0)$&	$(0)$&	$(0)$&	$(-)$&	$(0)$&	$(-)$&	$(0)$&	$(1)$&	$(0)$&	$(1)$&	$(0)$&	$(1)$&	$(1)$&	$(1)$&	$(1)$&	$(1)$&	$(-)$&	$(1)$&	$(0)$&	$(1)$&	$(-)$&	$(1)$&	\\
S&	$(-)$&	$(1)$&	$(1)$&	$(1)$&	$(0)$&	$(-)$&	$(-)$&	$(0)$&	$(-)$&	$(1)$&	$(-)$&	$(1)$&	$(1)$&	$(0)$&	$(1)$&	$(1)$&	$(1)$&	$(1)$&	$(1)$&	$(-)$&	$(1)$&	$(-)$&	$(1)$&	$(-)$&	$(1)$&	\\
T&	$(-)$&	$(1)$&	$(1)$&	$(0)$&	$(0)$&	$(0)$&	$(1)$&	$(1)$&	$(-)$&	$(0)$&	$(1)$&	$(1)$&	${\color{red}(N)}$&	$(1)$&	$(1)$&	$(1)$&	$(1)$&	$(-)$&	$(-)$&	$(1)$&	$(1)$&	$(1)$&	$(0)$&	$(1)$&	$(1)$&	\\
U&	$(1)$&	$(0)$&	$(1)$&	$(0)$&	$(0)$&	$(-)$&	$(0)$&	$(1)$&	$(1)$&	$(0)$&	$(1)$&	$(-)$&	${\color{red}(N)}$&	$(1)$&	$(1)$&	$(1)$&	$(1)$&	$(1)$&	$(1)$&	$(1)$&	$(1)$&	$(1)$&	$(-)$&	$(1)$&	$(1)$&	\\
V&	$(1)$&	$(0)$&	$(0)$&	$(1)$&	$(1)$&	$(1)$&	$(1)$&	$(1)$&	$(1)$&	$(0)$&	$(0)$&	$(-)$&	$(1)$&	$(1)$&	$(1)$&	$(-)$&	$(0)$&	$(0)$&	$(-)$&	$(1)$&	$(1)$&	$(1)$&	$(1)$&	$(-)$&	$(1)$&	\\
W&	$(0)$&	$(1)$&	$(0)$&	$(1)$&	$(1)$&	$(-)$&	$(0)$&	$(0)$&	$(-)$&	$(1)$&	$(1)$&	$(1)$&	$(1)$&	$(0)$&	$(1)$&	$(1)$&	$(1)$&	$(1)$&	$(1)$&	$(0)$&	$(-)$&	$(1)$&	$(1)$&	$(-)$&	$(1)$&	\\
X&	$(-)$&	$(-)$&	$(1)$&	$(0)$&	$(1)$&	$(-)$&	$(-)$&	$(1)$&	$(-)$&	$(1)$&	$(1)$&	$(1)$&	$(0)$&	$(1)$&	$(1)$&	$(-)$&	$(1)$&	$(-)$&	$(-)$&	$(1)$&	$(1)$&	$(-)$&	$(-)$&	$(1)$&	$(1)$&	\\
Y&	$(1)$&	$(1)$&	$(1)$&	$(1)$&	$(1)$&	$(1)$&	$(1)$&	$(1)$&	$(1)$&	$(1)$&	$(1)$&	$(1)$&	$(1)$&	$(1)$&	$(1)$&	$(1)$&	$(1)$&	$(1)$&	$(1)$&	$(1)$&	$(1)$&	$(1)$&	$(1)$&	$(1)$&	$(1)$&	\\
    \end{tabular}

\end{table}
\end{landscape}

Table~\ref{table:result} summarises the number of graphs output by the programs, when generating twin-free supergraphs of the six graphs in $Obs$ (see \autoref{f:graph_conj}). One can check that they obtain the very same numbers. 
   \begin{table}[h]
   \centering
    \begin{tabular}{|p{10ex}|p{25ex}|}
    \hline
     \center{Order\quad\quad\quad} & Number of twin-free supergraphs of the graphs in $Obs$ (\autoref{f:graph_conj}) 
    \\
\hline
6 & 1
\\
\hline
7 & 8
\\
\hline
8 & 34
\\
\hline
9 & 147
\\
\hline
10 & 539
\\
\hline
11 & 1649
\\
\hline
12 & 4535
\\
\hline
13 & 12281
\\
\hline
14 & 34207
\\
\hline
15 & 98591
\\
\hline
16 & 293621
\\
\hline
17 & 900956
\\
\hline
18 & 2839204
\\
\hline
19 & 9165983
\\
\hline
20 & 30263628
\\
\hline
21 & 102099119
\\
\hline
22 & 351834626
\\
\hline
    \end{tabular}
    \caption{Number of graphs output by the programs}\label{table:number-graphs-gen}
    \end{table}

 \end{document}